\documentclass[preprint,12pt,times]{elsarticle}
\usepackage{amsmath,amssymb,amsthm,mathtools,bm}
\usepackage{booktabs,array,tabularx}
\usepackage{microtype}
\usepackage{enumitem}
\usepackage{graphicx}
\usepackage{subfig}
\usepackage{caption}
\usepackage{placeins}
\usepackage{xcolor}
\usepackage[hidelinks,hypertexnames=false]{hyperref}
\usepackage[nameinlink,noabbrev]{cleveref}

\setlist{nosep}
\newtheorem{theorem}{Theorem}

\newtheorem{proposition}[theorem]{Proposition}
\newtheorem{corollary}[theorem]{Corollary}
\theoremstyle{remark}
\newtheorem{remark}{Remark}
\newtheorem{assumption}{Assumption}

\journal{Journal of Process Control}

\begin{document}

\begin{frontmatter}

\title{A Theory-Guided Advanced Regulatory Control Synthesis for Cooling-Limited Exothermic Semi-Batch Reactors}

\author[label1]{Chenchen Zhou}
\ead{chenchen.zhou@kuleuven.be}
\author[label1]{Jose Matias}
\ead{jose.matias@kuleuven.be}

\affiliation[label1]{organization={Chemical and Biochemical Reactor Engineering and Safety (CREaS), KU Leuven},
            addressline={Jan Pieter de Nayerlaan 5},
            city={Sint-Katelijne-Waver},
            postcode={2860},
            country={Belgium}}

\begin{abstract}
Cooling-limited exothermic semi-batch reactors require coordinated feed and cooling control to shorten batch time while maintaining the prescribed temperature. We develop a theoretical basis for advanced regulatory control (ARC) design by combining minimum-time and local safety analyses. Minimum-time analysis leads to an economic valve position control structure that adjusts feed using the temperature control system's cooling request, while cooling regulates temperature. Local safety analysis specifies the controller form and tuning conditions for reducing feed during cooling overload and restoring it as capacity becomes available. We also provide guidelines for industrial implementation and tuning. A reduced benchmark verifies the analytical tuning conditions, and an industrial-scale polymerization model evaluates the design. In simulations with parameter mismatch and unmodeled reaction dynamics, ARC achieves batch times comparable to those of parameter adaptive nonlinear model predictive control, using regulatory feedback without online nonlinear optimization.
\end{abstract}

\begin{highlights}
\item A theory-guided ARC design is proposed for cooling-limited semi-batch reactors.
\item Optimality analysis identifies virtual cooling demand as the economic feedback.
\item Local safety analysis gives controller form and tuning conditions.
\item ARC achieves batch times comparable to adaptive NMPC in an industrial reactor model.
\item Industrial guidelines connect the analysis to implementation and tuning.
\end{highlights}

\begin{keyword}
Advanced regulatory control \sep Semi-batch reactors \sep Thermal safety \sep Optimal control
\end{keyword}

\end{frontmatter}

\section{Introduction}\label{sec:intro}
Exothermic semi-batch reactors must balance productivity and thermal safety. Thermal runaway begins when heat generation exceeds the removal capacity of the reactor and utility system, causing temperature and pressure to rise uncontrollably; major accidents have repeatedly been linked to loss of cooling or poorly characterized reactive hazards \cite{Kummer2021,andriani2025}. Semi-batch operation is especially difficult because reactants accumulate, heat release changes with the batch state, and productivity is often improved by moving toward thermally active conditions \cite{ni2016,stoessel2021,kummer2020a,soroush1992,lu2017batch}. The controller must therefore maintain temperature while using the available thermal capacity to improve productivity and preserve product quality.

Industrial plants usually address this task with advanced regulatory control (ARC), which coordinates conventional controllers within the plant automation system~\cite{samad2020a,Wade2004ARC,Skogestad2023ARC,Gous2023ARCLevel}. ARC is transparent and maintainable, but exothermic semi-batch polymerization makes its design challenging. Startup, reaction, and finishing phases cannot be reliably separated by a fixed time schedule; their boundaries depend on reactant inventories, heat release, gas inventory, cooling availability, and model--plant mismatch. A loop architecture or pairing that is reasonable in one phase can become ineffective or unsafe later. Switching selectors or supervisory layers can also move the thermal load to a different actuator while integrators and saturated elements retain the memory from previous control configurations. The design problem is therefore architectural as well as parametric: controlled variables (CVs), manipulated variables (MVs), CV--MV pairings, controller elements, and tuning rules must be chosen as a coupled design.

Nonlinear model predictive control (NMPC) provides a systematic approach because it embeds the model, objectives, and constraints in an online optimal control problem~\cite{Rawlings2017MPCbook}; it has been demonstrated for temperature tracking and endpoint property control in semi-batch and polymerization applications~\cite{Nagy2005NMPCBatch,kummer2020a,Lucia2013SemiBatchNMPC}. Its practical burden remains high: nonlinear state estimation, repeated online optimization, model-mismatch sensitivity, and lifecycle maintenance~\cite{Nikolaou2001MPCNeeds,Skogestad2023ARC}. For industrial batch units with changing recipes or limited identification data, maintaining this predictive control system can be harder than maintaining an ARC architecture.

In semi-batch polymerization, joint process and control design treats nominal batch operation and control structure as a coupled problem~\cite{Alvarez2004Joint}. The operating trajectory, controlled and manipulated variables, and tracking controller are selected together. For exothermic batch reactors, joint operation and control design combines economically optimal operation with output feedback tracking and termination control, including stability conditions and industrial tuning rules~\cite{Andrade2025Joint}. Self-optimizing control and measurement-based dynamic optimization provide principles for selecting economic feedback variables~\cite{Skogestad2000SelfOptimizing,bonvin1998,srinivasanDynamicOptimizationBatch2003,srinivasanDynamicOptimizationBatch2003a}. In batch processes, optimality-based feedback has been used to regulate active constraints~\cite{Bonvin2003ActiveConstraints} and implement near-optimal temperature policies~\cite{Ferreira2025Feedback}. Nonlinear temperature control has also been developed with explicit treatment of jacket dynamics~\cite{Nguyen2023Jacket}.

This paper focuses on ARC design for cooling-limited exothermic semi-batch reactors. Increasing feed can shorten the batch, but the cooling system must remove the resulting reaction heat to maintain the required temperature. The design question is how to translate this cooling constraint into a feedback signal, controller structure and tuning conditions for feed adjustment. We use minimum-time and local safety analyses to guide these choices.

The main contributions are:
\begin{enumerate}
  \item \textbf{Minimum-time feed allocation.}
  Under the stated assumptions, we characterize time-optimal feed allocation with unused cooling capacity and at the cooling capacity limit.

  \item \textbf{Optimization-derived ARC architecture.}
  We derive a valve position control (VPC) structure that adjusts feed using the cooling request calculated by the temperature control system. This request is termed virtual cooling demand.

  \item \textbf{Safety-oriented controller form and tuning.}
  We use local safety analysis to specify the controller form and tuning conditions for feed reduction during cooling overload and recovery when capacity becomes available.
\end{enumerate}
We evaluate the procedure on a reduced isothermal benchmark and an industrial semi-batch polymerization case study \cite{li2025benchmark,zhou2026RealtimeNMPC}, including nominal operation, parameter mismatch and a process fault.

The rest of the paper is organized as follows. Section~\ref{sec:Preliminaries} states the reactor model and control problem. Section~\ref{sec:boundary-optimality} analyzes optimal feed allocation, and Section~\ref{sec:arc-structure} constructs the ARC structure. Sections~\ref{sec:safety-analysis} and~\ref{subsec:design-implications} develop the controller design conditions and industrial tuning procedure. Sections~\ref{sec:reduced-benchmark-sec} and~\ref{sec:industrial-case} evaluate the design on a reduced benchmark and an industrial polymerization model. Section~\ref{sec:discussion} discusses the results and concludes the paper.

\section{Problem Formulation and Background}
\label{sec:Preliminaries}

\subsection{Process Model}
We consider an open, nonisothermal semi-batch reactor without an outlet stream, described by mass and energy balances for liquid phase reactions \cite{rodrigues2015,bonvin1998}. The state consists of species inventories and reactor temperature,
$\mathbf n(t)\in\mathbb R^{n_s}$, $T(t)\in\mathbb R$, where $n_s$ is the number of species, $\mathbf n$ collects the species inventories, and $T$ is the reactor temperature. 
The dynamics are described by the mass and energy balances
\begin{align}
\dot{\mathbf n}(t) &= \mathbf N \ \mathbf r\bigl(\mathbf n(t),T(t)\bigr) \;+\; \mathbf W_{\mathrm{in}}\,\mathbf u_{\mathrm{in}}(t),
\label{eq:mass-balance}\\
C_P\!\bigl(\mathbf n(t),T(t)\bigr)\,\dot T(t)
&= G\!\bigl(\mathbf n(t),T(t)\bigr) + H\!\bigl(T(t),\mathbf u_{\mathrm{in}}(t)\bigr) - q_{\mathrm{ex}}\!\bigl(T(t),F_{\mathrm{cw}}(t)\bigr).
\label{eq:energy-balance}
\end{align}
Here $t$ denotes time, $\mathbf r(\mathbf n,T)\in\mathbb R^{n_r}$ is the reaction rate vector, and $\mathbf N\in\mathbb R^{n_s\times n_r}$ encodes stoichiometry. 
The manipulated inlet feed vector $\mathbf u_{\mathrm{in}}(t)\in\mathbb R^{n_u}$ is constrained by equipment limits:
$
\mathbf 0 \le \mathbf u_{\mathrm{in}}(t) \le \mathbf u^{\max},
$
and the constant matrix $\mathbf W_{\mathrm{in}}\in\mathbb R^{n_s\times n_u}$ maps the chosen feed coordinates to species inventory rates, including the corresponding inlet compositions.

In \eqref{eq:energy-balance}, $C_P(\mathbf n,T)$ is the effective heat capacity, $G(\mathbf n,T)$ is the reaction heat release rate, $H(T,\mathbf u_{\mathrm{in}})$ is the inlet sensible heat flow, and $q_{\mathrm{ex}}(T,F_{\mathrm{cw}})$ is utility heat removal. The scalar $F_{\mathrm{cw}}(t)\in[0,F_{\mathrm{cw}}^{\max}]$ is the cooling water flow rate.

\subsection{Control task}\label{subsec:control-task}

The operating objective is to complete the specified charge-and-deplete campaign in minimum time while respecting its operating constraints. We model one campaign on a finite horizon $[0,t_f]$, with $t_f>0$ as a decision variable.

Let $w_P\in\mathbb R^{n_r}_{\ge 0}$ weight the selected product formation channels. The cumulative product formation surrogate along a trajectory is then
\begin{equation}\label{eq:Np-path}
N_P(t)\ :=\ \int_0^{t} w_P^{\top}\mathbf r\big(\mathbf n(\tau),T(\tau)\big)\,\mathrm d\tau.
\end{equation}
We require the campaign to deliver at least a prescribed amount $N_P^{\mathrm{req}}$ of the specified product,
\begin{equation}\label{eq:ct-Np-req}
N_P(t_f)\ \ge\ N_P^{\mathrm{req}}.
\end{equation}
The state $(\mathbf n(\cdot),T(\cdot))$ evolves according to \eqref{eq:mass-balance}--\eqref{eq:energy-balance}, and the admissible control inputs are the inlet feeds $\mathbf u_{\mathrm{in}}(t)$ and the cooling water flow rate $F_{\mathrm{cw}}(t)$, both assumed measurable and bounded.

The operating policy must satisfy three types of requirement over $[0,t_f]$: limits that hold during operation, the temperature recipe used for product quality, and conditions that determine when the batch is complete.

\medskip\noindent
\textbf{(i) Operating path constraints.}
Physical limits and operability envelopes constrain the manipulated feeds, the utility flow, and the reactor temperature:
\begin{equation}\label{eq:ct-bounds}
\begin{aligned}
\mathbf 0 \le \mathbf u_{\mathrm{in}}(t) \le \mathbf u^{\max},\quad
0\le F_{\mathrm{cw}}(t)\le F_{\mathrm{cw}}^{\max},\quad
T^{\min}\le T(t)\le T^{\mathrm{safe}}_{\max},\quad t\in[0,t_f].
\end{aligned}
\end{equation}
Together these path constraints define the admissible operating envelope by enforcing actuator limits, cooling capacity limits, and a thermal safety upper bound.

\medskip\noindent
\textbf{(ii) Temperature specification and product quality.}
Product release can depend on terminal composition, molecular weight distribution, viscosity and the complete thermal history. These requirements, denoted abstractly by $\mathcal Q_{\mathrm{prop}}$, must be supplied for the intended grade. Reactor temperature tracking is the online enforceable quality surrogate in this control task. On the declared temperature hold interval $\mathcal I_{\rm hold}$, it must satisfy the recipe specification
\begin{equation}\label{eq:ct-tight-band-ref}
 |T(t)-T_{\mathrm{sp}}(t)|\le\Delta_T,\qquad t\in\mathcal I_{\rm hold},
\end{equation}
while additional grade specific properties still govern final product acceptance. Incorporating any tighter recipe limit gives an effective upper temperature bound. We consider operating intervals on which the lower temperature limit and additional property constraints do not restrict the local temperature and feed adjustments. The relevant thermal constraint is then
\begin{equation}\label{eq:ct-temp-cap}
T(t)\ \le\ T_{\mathrm{ref}}(t),\qquad t\in[0,t_f].
\end{equation}
The upper bound $T_{\rm ref}$ is distinct from the tracking setpoint $T_{\rm sp}$, which can be chosen below it to leave a temperature margin.

\medskip\noindent
\textbf{(iii) Endpoint logic.}
The batch is complete only after the scheduled material has been charged and the remaining reactive inventory is low enough that continuing reaction is negligible. In an ideal model this combines the production target $N_P(t_f)\ge N_P^{\mathrm{req}}$ in \eqref{eq:ct-Np-req} with a conversion set
\begin{equation}\label{eq:ct-conv}
\mathbf n(t_f) \ \in\ \mathcal C_{\mathrm{conv}},
\end{equation}
encoding targets on overall conversion and limits on residual monomer, initiator, or other species. Because these compositional criteria are not generally measured online, plant operation relies on endpoint signals indicating ``material charged'' and ``reactive material essentially depleted''.

Let $\mathcal D_C\subseteq\{1,\ldots,n_u\}$ contain the inlet streams with prescribed total charges. Their charge completion conditions are
\begin{equation}\label{eq:ct-dose}
 \int_0^{t_f}u_{i,\mathrm{in}}(t)\,dt=U_i^{\max},\qquad i\in\mathcal D_C.
\end{equation}
Other streams may have only a dose cap, $\int_0^{t_f}u_{i,\mathrm{in}}\,dt\le U_i^{\max}$.

For analysis, the depleted endpoint is written in terms of the residual reactant inventory
\begin{equation}\label{eq:residual-reactant}
n_R(t):=w_R^\top \mathbf n(t),
\end{equation} 
where $w_R \ge 0$ aggregates the limiting reactants (e.g., monomer). The depletion requirement is
\begin{equation}\label{eq:ct-conv-monomer}
  n_R(t_f)\le n_R^{\min}.
\end{equation}
In practice, this condition requires a measurable endpoint indicator that has been validated against residual reactant depletion for the process under consideration.

Together, \eqref{eq:ct-dose} and \eqref{eq:ct-conv-monomer} encode the charge-and-deplete endpoint logic used as the model level surrogate for the ideal production and conversion requirements \eqref{eq:ct-Np-req} and \eqref{eq:ct-conv}.

\subsection{Optimal control formulations}

The batch optimization problem treats the effective temperature bound \eqref{eq:ct-temp-cap}, dose \eqref{eq:ct-dose}, and depletion condition \eqref{eq:ct-conv-monomer} as hard constraints. The objective is to minimize the batch time required to satisfy them:
\begin{equation}\label{prob:ct-industrial}
\begin{aligned}
\min_{\mathbf u_{\mathrm{in}}(\cdot),\,F_{\mathrm{cw}}(\cdot),\,t_f}\quad & t_f \\[1mm]
\text{s.t.}\quad & \text{dynamics \eqref{eq:mass-balance}--\eqref{eq:energy-balance}},\quad t\in[0,t_f],\\
& \mathbf 0 \le \mathbf u_{\mathrm{in}}(t) \le \mathbf u^{\max},\\
& 0\le F_{\mathrm{cw}}(t)\le F_{\mathrm{cw}}^{\max},\\
& T^{\min}\le T(t)\le T^{\mathrm{safe}}_{\max},\\
& T(t) \le T_{\mathrm{ref}}(t),\qquad t\in[0,\,t_f],\\
& \int_0^{t_f}u_{i,\mathrm{in}}(t)\,dt=U_i^{\max}\ (i\in\mathcal D_C),\\
& \int_0^{t_f}u_{i,\mathrm{in}}(t)\,dt\le U_i^{\max}\ (i\notin\mathcal D_C),\\
& n_R(t_f)\ \le\ n_R^{\min}
\end{aligned}
\end{equation}

\begin{remark}[Generality of the Formulation]
The dynamic optimization problem \eqref{prob:ct-industrial} captures a structure common to many safety-critical semi-batch reactions. It follows the measurement-based optimization viewpoint of Srinivasan et al. \cite{srinivasanDynamicOptimizationBatch2003a,srinivasanDynamicOptimizationBatch2003}: active path constraints, here temperature and cooling capacity, determine the economically relevant operating policy under uncertainty.
\end{remark}

\subsection{ARC Synthesis Problem}

The optimal control problem above defines the regulatory control design problem considered in this paper. Following the self-optimizing control viewpoint \cite{Skogestad2000SelfOptimizing,Skogestad2004CSD}, the first design question is to choose a controlled variable whose regulation at a constant setpoint gives near-optimal operation despite disturbances and model mismatch. For the present semi-batch task, the synthesis must also decide the manipulated variable pairing, the ARC architecture elements needed to respect the path and endpoint constraints in \eqref{eq:ct-bounds}--\eqref{eq:ct-conv-monomer}, and the controller parameters, including setpoints, gains, and integral bounds.

The feedback variables are to be chosen from measurements and controller signals available in the plant. These include reactor temperature, pressure when a gas phase is present, feed and cooling flow rates, and controller outputs. The design must use these signals to adjust feed while respecting the prescribed temperature and charge requirements.

\section{Optimality Analysis for ARC Synthesis}
\label{sec:boundary-optimality}

To choose the economic feedback variable, we first ask when reducing cooling or reallocating feed can shorten a feasible batch. The following assumptions specify when these local changes preserve the operating limits and batch completion requirements.

\subsection{Assumptions for the optimality analysis}
\label{subsec:optimality-assumptions}

\begin{assumption}[Regular process and heat removal maps]\label{ass:regularity}
   \quad
   \begin{enumerate}
      \item The maps $\mathbf r$, $C_P$, $G$, $H$, and $q_{\mathrm{ex}}$ are $C^1$, with $C_P(\mathbf n,T)\ge C_{\min}>0$, $G(\mathbf n,T)\ge 0$, and $H(T,\mathbf u_{\mathrm{in}})$ affine in $\mathbf u_{\mathrm{in}}$ at fixed $T$.
      \item The heat removal map $q_{\mathrm{ex}}(T,F_{\mathrm{cw}})$ satisfies $\partial q_{\mathrm{ex}}/\partial F_{\mathrm{cw}}>0$ on the operating region and $\partial q_{\mathrm{ex}}/\partial T>0$ wherever cooling is active ($F_{\mathrm{cw}}>0$).
   \end{enumerate}
\end{assumption}

For jacket or coil cooling, increasing utility flow or the temperature difference between the reactor and coolant increases heat removal on an actively cooled branch \cite{srinivasanDynamicOptimizationBatch2003,srinivasanDynamicOptimizationBatch2003a,hungerbuhlerThermalProcessSafety2021,sneeCharacterisationExothermicReaction1992}. Boiling cooling requires a separate heat removal parametrization.

\begin{assumption}[Feasibility of the industrial campaign]\label{ass:feasibility}
There exists a nonempty set of initial conditions $\mathcal X_0$ such that, for every $(\mathbf n^0,T^0)\in\mathcal X_0$, problem~\eqref{prob:ct-industrial} admits at least one feasible solution satisfying the dynamics and constraints \eqref{eq:ct-bounds}--\eqref{eq:ct-conv-monomer}.
\end{assumption}

This restricts the analysis to initial conditions for which at least one recipe can deliver the prescribed charges and terminal conversion within the temperature and equipment limits.

\begin{assumption}[Positive thermal depletion sensitivity]\label{ass:depletion}
The reactant depletion rate $\chi(\mathbf n,T):=-w_R^\top \mathbf N\,\mathbf r(\mathbf n,T)$ is $C^1$ and strictly increasing in $T$ for fixed $\mathbf n$. On a cooled interior interval, there exists a bounded nonnegative cooling reduction profile with compact support of positive measure and uniform temperature and cooling flow margins. Its sufficiently small scaling at fixed feed preserves feasibility in the local trajectory tube over the remaining campaign and increases depletion integrated to the original terminal time to first order.
\end{assumption}
Arrhenius kinetics motivate the instantaneous temperature sensitivity~\cite{smithTemperaturedependenceElementaryReaction2008}. The terminal sensitivity also accounts for subsequent changes in reactant inventories; Appendix~\ref{sec:supp-proof-temp-boundary} gives a variational check.

\begin{assumption}[Feasible local feed retiming]\label{ass:dose-slack}
For sufficiently small horizon reductions $\tau>0$, let $\Delta U_i(\tau)=\int_{t_f^\star-\tau}^{t_f^\star}u_{i,\mathrm{in}}^\star(t)\,dt$. Each prescribed-dose component $i\in\mathcal D_C$ with $\Delta U_i(\tau)>0$ has an earlier set of positive measure $\mathcal J_i$ with uniform input slack $\delta_u>0$ on which the missing dose can be restored while preserving path constraints in the local trajectory tube. A component whose feed has already ended before the removed terminal interval needs no adjustment and no slack set.
\end{assumption}

This includes charging at the maximum rate followed by a conversion or holding phase without feed, since shortening the latter removes no prescribed dose \cite{saez-pardoModelingSolutionPolymerization2024}.

\subsection{Active cooling below the temperature bound}

We first consider intervals on which the cooling input is positive while the temperature constraint is inactive, i.e., $T(t)<T_{\mathrm{ref}}(t)$.

\begin{proposition}[No sustained active cooling below the temperature bound]\label{prop:temp-boundary}
Under Assumptions~\ref{ass:regularity}--\ref{ass:dose-slack}, an optimal solution of problem~\eqref{prob:ct-industrial} cannot contain an interval of positive measure $\mathcal I\subset[0,t_f^\star ]$ on which $T^\star (t)<T_{\mathrm{ref}}(t)$ while $F_{\mathrm{cw}}^\star (t)>0$.
\end{proposition}
The proof is given in Appendix~\ref{sec:supp-proof-temp-boundary}.

Proposition~\ref{prop:temp-boundary} gives a reason to operate at the prescribed temperature bound during actively cooled production: maintaining a lower temperature leaves an opportunity to shorten the batch under the stated assumptions. This assigns temperature tracking to the primary loop. Once that loop maintains $T_{\mathrm{ref}}$, the remaining economic decision is how to adjust feed within the available cooling capacity.

\phantomsection\label{par:tracking-correction}
To determine feed allocation at the temperature bound, we now impose exact tracking, $T(t)\equiv T_{\mathrm{ref}}(t)$. The required heat removal rate is
\begin{equation}\label{eq:qreq-tv}
q_{\mathrm{req}}(\mathbf n,\mathbf u,t)
:= G(\mathbf n,T_{\mathrm{ref}}(t))+H(T_{\mathrm{ref}}(t),\mathbf u)
   - C_P(\mathbf n,T_{\mathrm{ref}}(t))\,\dot T_{\mathrm{ref}}(t).
\end{equation}
The cooling flow required to provide this heat removal is
\begin{equation}\label{eq:Freq-tv}
F_{\mathrm{cw}}^{\mathrm{th}}(\mathbf n,\mathbf u,t)
= q_{\mathrm{ex}}^{-1}\!\Bigl(T_{\mathrm{ref}}(t),\, q_{\mathrm{req}}(\mathbf n,\mathbf u,t)\Bigr),
\end{equation}
where $q_{\mathrm{ex}}^{-1}(T,\cdot)$ denotes the cooling flow value that achieves a specified heat removal rate at temperature $T$. This inverse is well defined on its range by the monotonicity in Assumption~\ref{ass:regularity}.

The cooling capacity at the temperature cap is $q_{\mathrm{ex}}(T_{\mathrm{ref}}(t),F_{\mathrm{cw}}^{\max})$. Define the cooling capacity residual as required heat removal minus maximum heat removal:
\begin{equation}\label{eq:phi-tv}
\phi(\mathbf n,\mathbf u,t)
:= q_{\mathrm{req}}(\mathbf n,\mathbf u,t)
   - q_{\mathrm{ex}}(T_{\mathrm{ref}}(t),F_{\mathrm{cw}}^{\max})\ \le\ 0.
\end{equation}
Here $\phi<0$ means that heat removal capacity remains, $\phi=0$ is the active cooling boundary, and $\phi>0$ means that tracking $T_{\mathrm{ref}}(t)$ would require more cooling than the plant can provide.
The feasible feed rates at a given state and time therefore form the set
\begin{equation}\label{eq:UT-tv}
\mathcal U_T(\mathbf n,t):=\Big\{\mathbf u\in\mathbb R^{n_u}\ \big|\ \mathbf 0 \le \mathbf u \le \mathbf u^{\max},\ \phi(\mathbf n,\mathbf u,t)\le 0\Big\}.
\end{equation}
On the actively cooled tracking branch, the lower cooling bound is locally inactive: $q_{\mathrm{req}}>q_{\mathrm{ex}}(T_{\mathrm{ref}},0)$. Outside this branch its lower-bound inequality must also be retained. Because $H(T_{\mathrm{ref}}(t),\mathbf u)$ is affine in feed, $\phi(\mathbf n,\mathbf u,t)$ is affine in $\mathbf u$ for fixed state and time. Hence $\mathcal U_T(\mathbf n, t)$ is the intersection of a box and a half-space: a possibly empty compact convex polytope. It is nonempty whenever $\phi(\mathbf n,\mathbf 0,t)\le 0$.

To write the total feed dosage constraint \eqref{eq:ct-dose} as an endpoint constraint, introduce the cumulative feed amount $\mathbf m(t)\in\mathbb R^{n_u}$:
\begin{equation}\label{eq:dose-dyn}
\dot{\mathbf m}(t) = \mathbf u_{\mathrm{in}}(t), \quad \mathbf m(0) = \mathbf 0,
\end{equation}
with $m_i(t_f)=U_i^{\max}$ for $i\in\mathcal D_C$ and $m_i(t_f)\le U_i^{\max}$ otherwise. Nonnegative feed then implies $\mathbf m(t)\le\mathbf U^{\max}$ throughout the batch.

The reduced minimum-time problem is
\begin{equation}\label{prob:prod-tracking}
\begin{aligned}
\min_{\mathbf u_{\mathrm{in}}(\cdot), t_f}\quad & J = t_f \\[1mm]
\text{s.t.}\quad &
\dot{\mathbf n}(t)=\mathbf N\,\mathbf r\!\bigl(\mathbf n(t),T_{\mathrm{ref}}(t)\bigr)+\mathbf W_{\mathrm{in}}\,\mathbf u_{\mathrm{in}}(t),\\
& \dot{\mathbf m}(t) = \mathbf u_{\mathrm{in}}(t),\\
& \mathbf u_{\mathrm{in}}(t)\in\mathcal U_T\bigl(\mathbf n(t),t\bigr),\quad t\in [0,t_f],\\
& \mathbf m(0)=\mathbf 0,\quad m_i(t_f)=U_i^{\max}\ (i\in\mathcal D_C),\\
& m_i(t_f)\le U_i^{\max}\ (i\notin\mathcal D_C),\\
& n_R(t_f) \le\ n_R^{\min}.
\end{aligned}
\end{equation}

\subsection{Hamiltonian Feed Allocation on the Tracking Manifold}

The Hamiltonian associated with the minimum-time problem \eqref{prob:prod-tracking} separates the reaction-dependent terms from the input-affine part:
\begin{equation}\label{eq:Hamiltonian-tv}
\begin{aligned}
\mathcal{H}(\mathbf n,\mathbf m,\mathbf u,\boldsymbol{\lambda},t)
&= \lambda_c + \boldsymbol{\lambda}_n^\top \mathbf N\,\mathbf r\bigl(\mathbf n,T_{\mathrm{ref}}(t)\bigr)
+ \bigl(\mathbf W_{\mathrm{in}}^\top \boldsymbol{\lambda}_n+\boldsymbol{\lambda}_m\bigr)^\top \mathbf u\\
&=: \mathcal H_0(\mathbf n,\boldsymbol{\lambda},t)
 + \boldsymbol{\sigma}(\mathbf n,\boldsymbol{\lambda},t)^\top\mathbf u,
\end{aligned}
\end{equation}
where $\lambda_c\ge0$ is the cost multiplier (normalized to one for a normal extremal), $\boldsymbol{\lambda} \in \mathbb R^{n_s+n_u}$ is the costate vector, and $\boldsymbol{\sigma}(\cdot)$ is the \emph{switching function} vector.
For an extremal satisfying Pontryagin's minimum principle on the regular tracking branch, the optimal control $\mathbf u^\star(t)$ minimizes this Hamiltonian pointwise:
\begin{equation}\label{eq:ham-min-tv}
\mathbf u^\star(t) \in \arg\min_{\mathbf u \in \mathcal U_T(\mathbf n^\star(t), t)} \boldsymbol{\sigma}(t)^\top \mathbf u.
\end{equation}
\begin{proposition}[Active-constraint structure of the time-optimal feed]\label{prop:boundary-tv}
Under Assumption~\ref{ass:regularity} and the pointwise minimum-principle condition, the Hamiltonian minimization \eqref{eq:ham-min-tv} is a linear program over the convex polytope $\mathcal U_T(\mathbf n, t)$. The optimal feed rate $\mathbf u^\star(t)$ is determined by the active constraints:
\begin{itemize}
    \item \textbf{Cooling margin regime.} If cooling margin remains,
    \[
    \phi(\mathbf n^\star(t),\mathbf u^\star(t),t)<0,
    \]
    the feed rate is determined solely by the switching function signs:
    \[
    u_i^\star(t) = \begin{cases}
      0 & \text{if } \sigma_i(t) > 0, \\
      u_i^{\max} & \text{if } \sigma_i(t) < 0, \\
      \text{singular} & \text{if } \sigma_i(t) = 0.
    \end{cases}
    \]
    \item \textbf{Cooling-limited regime.} If the cooling margin is exhausted on the tracking manifold,
      \[ \phi(\mathbf n^\star(t), \mathbf u^\star(t), t) = 0, \]
    then the feed must satisfy the active cooling capacity equality. Equivalently, $F_{\mathrm{cw}}^{\mathrm{th}} = F_{\mathrm{cw}}^{\max}$, so the required cooling reaches the maximum available cooling capacity.
\end{itemize}
\end{proposition}
The proof is given in Appendix~\ref{sec:supp-proof-boundary-tv}.

The cooling-limited case in Proposition~\ref{prop:boundary-tv} identifies a control objective: adjust feed so that the cooling required to maintain temperature stays near the available capacity.

\begin{remark}[Typical three-stage operation]\label{rem:three-stage}
A typical semi-batch polymerization passes through three operating stages. Early in the reaction, heat generation is modest and cooling margin remains ($\phi<0$), so the switching function signs and feed bounds determine the feed rate. During the main thermally limited period, heat generation reaches the cooling limit and the feed must be throttled so that $\phi=0$, equivalently $F_{\mathrm{cw}}^{\mathrm{th}}=F_{\mathrm{cw}}^{\max}$. Near the endpoint, dose completion, residual reactant depletion, or finishing constraints dominate, and the feed is terminated or moved to its finishing policy.
\end{remark}

\section{ARC Architecture from Cooling Demand Regulation}
\label{sec:arc-structure}
The temperature controller adjusts cooling to track the prescribed temperature. Its cooling request also indicates when feed should increase or decrease. We use this request as the feedback signal for feed control.

\subsection{Virtual cooling demand and actuator saturation}
\label{subsec:virtual-cooling}

The cooling flow stops increasing when the actuator saturates, even if the temperature controller requests more cooling. To retain this information, define the \emph{virtual cooling demand} $F_{\mathrm{cw}}^{\mathrm{v}}(t)$ as the cooling water request calculated by the temperature control system without applying feasibility limits or physical saturation:
\begin{equation}
  F_{\mathrm{cw}}^{\mathrm{v}}(t) = C_T\big(T_{\mathrm{ref}}(t) - T(t)\big),
  \label{eq:virtual-cooling-demand}
\end{equation}
where $C_T(\cdot)$ represents the temperature control system, with larger $F_{\mathrm{cw}}^{\mathrm{v}}$ indicating a larger cooling request. It may be a single controller or a temperature cascade.

For ideal static actuation, the physical cooling water flow $F_{\mathrm{cw}}(t)$ is obtained from $F_{\mathrm{cw}}^{\mathrm{v}}(t)$ through an explicit saturation nonlinearity:
\begin{equation}
F_{\mathrm{cw}}(t)=\operatorname{sat}\!\big(F_{\mathrm{cw}}^{\mathrm{v}}(t);\,0,\,F_{\mathrm{cw}}^{\max}\big):=\max\bigl(0,\,\min\bigl(F_{\mathrm{cw}}^{\mathrm{v}}(t),F_{\mathrm{cw}}^{\max}\bigr)\bigr)
\label{eq:sat-block}
\end{equation}
Regulating the output of the temperature control system by adjusting feed gives a valve position control (VPC) structure. The economic VPC loop regulates virtual cooling demand near a setpoint $F_{\mathrm{cw,sp}}\approx F_{\mathrm{cw}}^{\max}$:
\begin{equation}
   \mathbf u_{\mathrm{in}}(t) = C_{\mathrm{ceiling}}\big(F_{\mathrm{cw,sp}} - F_{\mathrm{cw}}^{\mathrm{v}}(t)\big),
   \label{eq:ceiling-controller}
\end{equation}
where $C_{\mathrm{ceiling}}(\cdot)$ is the feed controller. It increases the feed request when cooling demand is below $F_{\mathrm{cw,sp}}$ and reduces it when demand is above the setpoint. Figure~\ref{fig:control-structure} shows the resulting two-loop structure.

\begin{figure}[!htbp]
\centering
\includegraphics[width=\columnwidth]{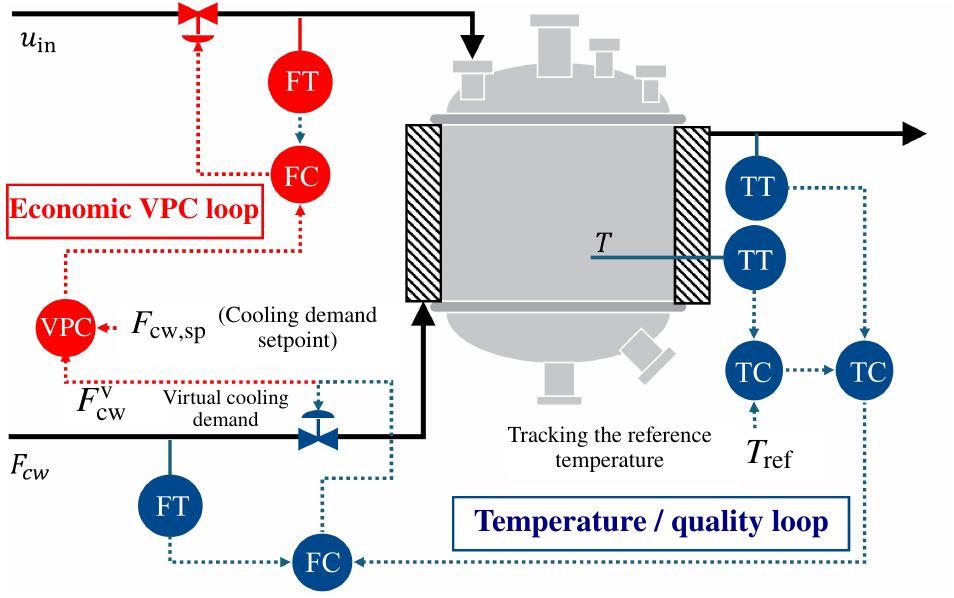}
\caption{ARC structure. TT/FT: temperature/flow transmitters; TC/FC: temperature/flow controllers; VPC: valve position controller for cooling demand regulation.}
\label{fig:control-structure}
\end{figure}
\FloatBarrier

\section{Safety-Oriented Controller Synthesis and Tuning}
\label{sec:safety-analysis}

When cooling demand is too high, feed must decrease promptly; when capacity becomes available, feed must recover to resume productive operation. We use a PI controller with a bounded integral state. Proportional action responds to the current cooling demand, while the integral bounds limit how much past error can delay feed reduction or recovery.

The feed response also depends on actuator dynamics and reactants already in the vessel. We first bound the feed requested by the PI controller, then use a local dynamic model to determine how feed affects the thermal margin. These bounds give conditions on the PI gains and integral limits.

\subsection{Reduced coordinate and error band}
\label{subsec:safety-signals}

Let $e_v$ be a scalar thermal margin error, oriented positively toward unused capacity. For the cooling demand VPC, it is the cooling demand error
\begin{equation}\label{eq:ev-def}
e_v(t) \;:=\; F_{\mathrm{cw,sp}} \;-\; F_{\mathrm{cw}}^{\mathrm{v}}(t),
\end{equation}
where $e_v>0$ means that virtual cooling demand is below its setpoint and $e_v<0$ means that demand exceeds it.

For cooling demand regulation, the analysis focuses on the near-boundary cooling capacity regime
\begin{equation}\label{eq:regime-def}
\mathcal R \;:=\; \bigl\{\, t \in [0, t_f] \;:\; |\phi(\mathbf n, \mathbf u, t)| \;\le\; \delta \,\bigr\},
\end{equation}
where $\delta>0$ sets the distance from the cooling capacity boundary. In this region, small feed adjustments can affect the balance between heat generation and cooling capacity. The goal is to keep the cooling demand error $e_v$ within the error band
\begin{equation}\label{eq:working-band}
\mathcal E \;:=\; [-\eta_s,\eta_e],
\qquad
\eta_s>0,\ \eta_e>0,
\end{equation}
where $-\eta_s$ bounds demand above its setpoint and $\eta_e$ bounds demand below it.

\subsection{Projected PI control and cumulative feed}
\label{subsec:structural-props}

For the reduced analysis, feed actuation follows one direction with fixed composition,
\begin{equation}\label{eq:uin-1d}
\mathbf u_{\mathrm{in}}(t)=\alpha(t)\,\mathbf d,
\qquad \alpha(t)\in[0,\alpha^{\max}],
\end{equation}
where $\mathbf d$ is a fixed feed direction and $\alpha$ changes feed intensity but not composition. The economic PI controller is
\begin{equation}\label{eq:vpc-law-1d}
\alpha(t)=\operatorname{sat}\!\bigl(\alpha_b+K_P\,e_v(t)+r(t);\,0,\,\alpha^{\max}\bigr),
\qquad K_P\ge 0,
\end{equation}
where $\alpha_b$ is the bias term and $r(t)$ is the integral state, evolving as
\begin{equation}\label{eq:proj-int}
\dot r(t)=
\begin{cases}
K_I\,e_v(t), & -r_-<r(t)<r_+,\\[1mm]
\min\{K_I\,e_v(t),\,0\}, & r(t)=r_+,\\[1mm]
\max\{K_I\,e_v(t),\,0\}, & r(t)=-r_-.
\end{cases}
\end{equation}
Here $K_I>0$ and $r(t_0)\in[-r_-,r_+]$. The projection keeps $r(t)$ within fixed bounds and thereby limits feed adjustments caused by integral action. For $K_P>0$, rearranging~\eqref{eq:vpc-law-1d} shows that $e_v\le-(r_++\alpha_b)/K_P$ makes the pre-saturation command nonpositive for every admissible $r(t)\le r_+$. Hence $\alpha=0$: proportional action overcomes even the largest permitted integral state and sets the feed request to zero. Actual feed reduction then depends on actuator response.

The initial feed rate alone does not determine the dose delivered over a window, because the error and integral state continue to change. We therefore construct upper and lower feed envelopes over a window of length $\tau$, starting from error $e_k$. Suppose $|\dot e_v|\le L_e$ on the window. The curves $E^\pm(s)=e_k\pm L_e s$ bound the possible error motion. Let $R^\pm$ solve the projected integrator~\eqref{eq:proj-int}, driven by $E^\pm$, with $R^+(0)=r_+$ and $R^-(0)=-r_-$. Starting from the extreme integral states and driving them with the extreme errors gives $R^-(s)\le r(t_k+s)\le R^+(s)$ by scalar comparison. Substitution into the PI law gives the feed envelopes
\begin{equation}\label{eq:alpha-s-max}
 a^+(s;e_k)=\operatorname{sat}\bigl(\alpha_b+K_PE^+(s)+R^+(s);0,\alpha^{\max}\bigr),
\end{equation}
\begin{equation}\label{eq:alpha-e-min}
 a^-(s;e_k)=\operatorname{sat}\bigl(\alpha_b+K_PE^-(s)+R^-(s);0,\alpha^{\max}\bigr).
\end{equation}
These scalar equations are evaluated offline from a declared error rate bound. Measured or model-derived feed envelopes can be used if they bound the whole window. For actual feed behind actuator dynamics, propagate the command envelopes through those dynamics before forming the feed budgets. Include mode changes and dose-stop overrides in the window model or end the test at their occurrence.

\subsection{Input-output order and local endpoint test}
\label{subsec:thermal-bridge}

We now relate feed to the thermal margin error through a local dynamic model. For cooling demand error, this model includes temperature control, actuator and measurement dynamics, so temperature tracking errors enter through the changing cooling request.

The input-output order is the lowest derivative of the error in which feed appears. A first-order relation can arise when feed directly changes the reactive inventory defining the error. An intermediate inventory can give a second-order relation when no direct path is present. We consider either order, with a bounded residual to account for model approximation.

\begin{assumption}[Verified local input-output model]\label{ass:curvature}
On a declared operating region and analysis window, the selected coordinate satisfies, for a specified order $m\in\{1,2\}$,
\begin{equation}\label{eq:curvature-model}
 e_v^{(m)}=d_0(t)-q(t)\alpha(t)+\rho_e(t),
\end{equation}
with
\begin{equation}\label{eq:curvature-bounds}
 d_-\le d_0\le d_+,\qquad 0<q_-\le q\le q_+,\qquad |\rho_e|\le\epsilon_\rho.
\end{equation}
The order, coefficient bounds and residual bound are verified for this coordinate, feed direction and operating region.
\end{assumption}
The drift $d_0$ is the error rate or acceleration without feed; the positive coefficient $q$ describes how feed reduces that rate or acceleration. The residual accounts for model error, with $\epsilon_\rho=0$ for an exact relation. The coefficient and residual bounds supply the process information needed to select the PI gains and integral limits.

Integrating a first-order relation makes the endpoint depend on the total feed over the window. Integrating a second-order relation twice gives earlier feed more weight and also retains the initial error slope. To write both cases together, set $k_1(s)=1$, $k_2(s)=\tau-s$, $c_1=\tau$, $c_2=\tau^2/2$, and $v_{1,k}=0$, $v_{2,k}=\tau\dot e_v(t_k)$. Direct integration gives
\begin{equation}\label{eq:window-law}
 e_{k+1}-e_k=v_{m,k}+\int_0^\tau k_m(s)[d_0-q\alpha+\rho_e](t_k+s)\,ds.
\end{equation}
The corresponding cumulative or time-weighted feed is
\begin{equation}\label{eq:Jw-def}
 J_{m,k}=\int_0^\tau k_m(s)\alpha(t_k+s)\,ds.
\end{equation}
At the two band boundaries, integrating the feed envelopes gives the upper and lower budgets
\begin{equation}\label{eq:Jw-boundary}
 J_{m,s}^+=\int_0^\tau k_m(s)a^+(s;-\eta_s)\,ds,\qquad
 J_{m,e}^-=\int_0^\tau k_m(s)a^-(s;\eta_e)\,ds.
\end{equation}
Since $k_m,q,\alpha\ge0$, the endpoint obeys
\begin{align}
 e_{k+1}&\ge e_k+v_{m,k}+c_m(d_--\epsilon_\rho)-q_+J_{m,k},\label{eq:window-lower}\\
 e_{k+1}&\le e_k+v_{m,k}+c_m(d_++\epsilon_\rho)-q_-J_{m,k}.\label{eq:window-upper}
\end{align}

\begin{theorem}[Local endpoint test]\label{thm:band-invariance}
Suppose Assumption~\ref{ass:curvature} holds throughout the window. At $e_k=-\eta_s$, let $v_{m,k}\ge v_{m,s}^-$; at $e_k=\eta_e$, let $v_{m,k}\le v_{m,e}^+$. For $m=1$ these bounds equal zero. For $m=2$ they are
\begin{equation}\label{eq:boundary-slopes}
 v_{2,s}^-=\tau\nu_s^-,\qquad v_{2,e}^+=\tau\nu_e^+,
\end{equation}
where $\nu_s^-$ and $\nu_e^+$ bound the respective initial error slopes. Assume the budgets, valid throughout the window, satisfy
\begin{equation}\label{eq:boundary-jw-bounds}
 e_k=-\eta_s\Rightarrow J_{m,k}\le J_{m,s}^+,\qquad
 e_k=\eta_e\Rightarrow J_{m,k}\ge J_{m,e}^-.
\end{equation}
If
\begin{equation}\label{eq:safety-window-cond}
 v_{m,s}^-+c_m(d_--\epsilon_\rho)-q_+J_{m,s}^+\ge0,
\end{equation}
\begin{equation}\label{eq:economic-window-cond}
 v_{m,e}^++c_m(d_++\epsilon_\rho)-q_-J_{m,e}^-\le0,
\end{equation}
then $e_k=-\eta_s$ implies $e_{k+1}\ge-\eta_s$, and $e_k=\eta_e$ implies $e_{k+1}\le\eta_e$.
\end{theorem}
The proof is given in Appendix~\ref{sec:proof-local-endpoint}.

Theorem~\ref{thm:band-invariance} gives two design checks. At the lower boundary, the upper feed budget must be small enough to keep the next endpoint at or above that boundary. At the upper boundary, the lower feed budget must be large enough to keep the next endpoint at or below it. For cooling demand regulation, this means reducing feed during excessive demand and restoring feed when capacity becomes available. Compute the PI feed envelopes and budgets, then check~\eqref{eq:safety-window-cond} and~\eqref{eq:economic-window-cond}. To keep the error inside the band throughout the window, the next result checks its derivative at both boundaries for first-order response.

\begin{corollary}[First-order inward-boundary test]\label{cor:first-order-band}
Suppose Assumption~\ref{ass:curvature} holds with $m=1$ on an operating region, the closed-loop solution is well posed, and the projected integral state starts in $[-r_-,r_+]$. Put
\[
 A_s^+=\operatorname{sat}(\alpha_b-K_P\eta_s+r_+;0,\alpha^{\max}),\qquad
 A_e^-=\operatorname{sat}(\alpha_b+K_P\eta_e-r_-;0,\alpha^{\max}).
\]
If the following inequalities hold on the corresponding boundary slices of that region,
\begin{equation}\label{eq:first-order-inward}
 d_0-q A_s^+-\epsilon_\rho\ge0\quad(e_v=-\eta_s),\qquad
 d_0-q A_e^-+\epsilon_\rho\le0\quad(e_v=\eta_e),
\end{equation}
then a trajectory starting in $\mathcal E$ remains in $\mathcal E$ until it leaves the declared operating region or the controller mode changes.
\end{corollary}
The proof is given in Appendix~\ref{sec:proof-first-order-inward}.

Corollary~\ref{cor:first-order-band} gives a direct test for keeping the error inside the band. The proportional gain and integral limits determine the largest feed at the lower error boundary and the smallest feed at the upper error boundary. They can therefore be selected by checking~\eqref{eq:first-order-inward} on both boundaries. The integral gain controls the response within the band and can be assessed through the finite-window budgets of Theorem~\ref{thm:band-invariance}. Checking $d_0-qA$ jointly on each boundary avoids the conservatism of unrelated coefficient extrema. The thermal meaning of the error band depends on the chosen coordinate, as described in Appendix~\ref{sec:supp-continuous-bridge}.

An additional feed cap $\alpha_Q\ge0$ gives $\alpha=\min\{\alpha_{\rm PI},\alpha_Q\}$. The feed upper bound $A_s^+$ at the lower error boundary remains valid; at the upper error boundary replace $A_e^-$ by $\min\{A_e^-,\alpha_Q\}$ and check the resulting state-dependent expression.

The same first-order argument also permits a bounded measurement error. If the PI uses $\tilde e_v=e_v+\xi$ with $|\xi|\le\epsilon_m$, replace $A_s^+$ by $\operatorname{sat}[\alpha_b+K_P(-\eta_s+\epsilon_m)+r_+;0,\alpha^{\max}]$ and $A_e^-$ by $\operatorname{sat}[\alpha_b+K_P(\eta_e-\epsilon_m)-r_-;0,\alpha^{\max}]$ in~\eqref{eq:first-order-inward}. Projection still bounds the integral state, so the proof is unchanged. These modified bounds allow a specified measurement error to be included when selecting controller parameters.

\section{Industrial Implementation and Tuning}
\label{subsec:design-implications}\label{par:design-information}
With cooling assigned to temperature control and feed to economic VPC, implementation starts by tuning the temperature loop. The economic PI controller is then configured and tuned with that loop in operation. The following steps use the prescribed temperature tolerance, feed limits, total charges and batch completion criteria.
\begin{enumerate}
\item \emph{Temperature control.} The temperature loop should keep reactor temperature within the prescribed tolerance whenever sufficient cooling capacity is available. The temperature controller may act directly on cooling flow. A reactor--jacket cascade is an alternative when jacket temperature can be measured and controlled faster than reactor temperature. The jacket loop is then tuned first, followed by the reactor temperature controller with the jacket loop closed. Local process models or response tests provide the gains and response times needed for tuning. SIMC is one established option for suitable first-order models with delay~\cite{Skogestad2003Tuning}.

\item \emph{Economic VPC.} Connect virtual cooling demand to the economic PI controller, retaining the signal when cooling saturates. Choose its setpoint below the cooling limit to leave a reserve. More reserve is appropriate when heat release is strongly sensitive to temperature, reactants accumulate, or cooling and feed respond slowly.

The controller can adjust a feed flow setpoint directly. For gas-fed reactions, it may instead adjust the setpoint of a pressure controller that manipulates gas feed; gas inventory and transfer into the reacting phase then affect the response. Each feed is limited to its allowable rate and stopped when its prescribed charge is reached.

\item \emph{Selection of the economic PI settings.} Tuning requires information about how a change in feed affects cooling demand. With the temperature loop active, records from an established feed recipe or permitted small feed changes can provide this information. The input changes must be sufficient to distinguish their effect from the batch's natural evolution. Different reaction stages may require different local models, as illustrated by piecewise affine identification of semi-batch reactors~\cite{Lassoued2014PWARX}. Each model should relate the selected feed input to the cooling demand error, including measurement filtering and any intervening feed or pressure control. The local model and its uncertainty bounds are used to select the PI gains and integral limits through the tuning conditions in Corollary~\ref{cor:first-order-band} or Theorem~\ref{thm:band-invariance}.

For specified changes in cooling demand, the economic PI control law and integral limits give bounds on the feed request over time. Including actuator dynamics then allows the actual feed reduction and recovery times to be evaluated. The required response times are selected from the available thermal margin and observed process response. These calculations can also guide parameter selection when a local model of the cooling demand response is unavailable. Gains and integral limits are chosen to meet the response requirements with some margin, rather than at the boundary of the admissible range.

\item \emph{Commissioning and adjustment.} Initial operation with economic VPC requires operator supervision and an available protective feed shutoff. Observations should cover startup, reaction, charge completion and finishing. If the feed request remains high during cooling capacity overload, proportional action may be too weak to overcome the accumulated integral state; the gain and upper integral limit are therefore checked together. If the request falls promptly but actual feed falls slowly, the actuator or pressure response needs attention. Repeated feed cutoffs and recoveries may indicate excessive proportional or integral action, or interaction with the temperature loop. Conversely, slow feed recovery while cooling capacity remains unused calls for checking the integral recovery rate, lower integral limit and feed constraints. Parameter adjustments are checked against the same response conditions used for the initial settings.
\end{enumerate}

\section{Reduced Verification Benchmark}
\label{sec:reduced-benchmark-sec}
\label{subsec:reduced-benchmark}

We verify the local PI tuning method on the reduced isothermal semi-batch benchmark of Srinivasan et al.~\cite{srinivasanDynamicOptimizationBatch2003a,srinivasanDynamicOptimizationBatch2003}. Reactant A is initially charged, reactant B is fed at volumetric rate $u$, and the reactor temperature is held constant. Thermal safety is specified by the cooling failure temperature $T_{\mathrm{cf}}$, the temperature reached if cooling fails and the remaining reactants react adiabatically. The benchmark objective is to maximize conversion at a fixed batch duration:
\begin{equation}\label{prob:isothermal-benchmark}
\begin{aligned}
 \max_{u(\cdot)}\quad &x_a(t_f),\\
 \dot x_a&=R_{\mathrm{iso}}V/N_A^0,\qquad \dot V=u,\\
 0\le u&\le u_{\max},\quad V\le V_{\max},\quad T_{\mathrm{cf}}\le T_{\max},\\
 \Delta H R_{\mathrm{iso}}V&\le q_{\mathrm{ex}}^{\max}(V),\qquad x_a(0)=0,\ V(0)=V_A^0.
\end{aligned}
\end{equation}
Here $x_a$ is the conversion of A, $V$ is reactor volume, and $N_A^0$ is the initial amount of A. The reactant inventories are $N_A=N_A^0(1-x_a)$ and $N_B=c_{B,f}(V-V_A^0)-N_A^0x_a$, giving the reaction rate $R_{\mathrm{iso}}=kN_AN_B/V^2$ and cooling failure temperature $T_{\mathrm{cf}}=T+C\min(N_A,N_B)/V$. Here $C=\Delta H/(\rho c_p)$, and $c_{B,f}$ is the concentration of B in the feed.

For comparison with the PI controller, we solve the optimal control problem (OCP)~\eqref{prob:isothermal-benchmark} using CasADi and IPOPT~\cite{andersson2019,biegler2009}. Details of the numerical solution and constraint checks are given in the supplementary material.

The controller uses the projected PI law in Section~\ref{subsec:structural-props}, with $\alpha=u$ and $\alpha_b=0$. With temperature held constant, it regulates the cooling failure temperature error $e_v=T_{\max}-T_{\mathrm{cf}}-e_{\mathrm{sp}}$. Thus $e_v=0$ retains a temperature margin $e_{\mathrm{sp}}>0$, while $e_v<0$ approaches the cooling failure boundary. On the B-limited branch $0<N_B\le N_A$, with $c_B=N_B/V$,
\begin{equation}\label{eq:reduced-first-order}
 \dot e_v=Ck\frac{N_A}{V}c_B-\frac{C(c_{B,f}-c_B)}{V}u
          =d_0-qu.
\end{equation}
Feed appears in the first derivative, with $q>0$ when $c_B<c_{B,f}$. We therefore use the first-order condition in Corollary~\ref{cor:first-order-band} to check the PI settings on this branch.

The PI controller and OCP are subject to the same two thermal constraints. To enforce the instantaneous cooling capacity limit, we apply a cooling capacity selector to the PI request. Define the cooling capacity margin $h_Q=Q_{\max}(V)-\Delta H R_{\rm mol}$, with $R_{\rm mol}=kN_AN_B/V$ and $Q_{\max}=q_{\rm ex}^{\max}$. Writing $-\dot h_Q=a_Q+b_Qu$, use
\begin{equation}\label{eq:heat-selector}
 u=\min\{u_{\rm PI},u_Q\},\qquad
 u_Q=\operatorname{sat}\!\left(\frac{-a_Q+\kappa_Qh_Q}{b_Q};0,u_{\max}\right)
\end{equation}
for $b_Q>0$, and $u_Q=u_{\max}$ otherwise; $\kappa_Q=0.2$ s$^{-1}$. Metered volume shutoff gives $u=0$ at $V_{\max}$.
The supplementary material gives the expressions for $a_Q,b_Q$ and the proof that the controller satisfies both thermal constraints and the volume limit under the stated initial and gain conditions. It also checks that the selector permits feed recovery.

Following Section~\ref{subsec:design-implications}, we select the temperature margin and error band, then check the PI settings at both boundaries. Table~\ref{tab:reduced-screen} gives the settings and the resulting margins. At the lower error boundary, the integral limit ensures zero feed. At the upper boundary, the PI request exceeds $d_0/q$, the feed needed to hold the error constant, and moves the error back into the band. With the selector included, the inward margins are $1.1191\times10^{-3}$ and $1.1094\times10^{-3}$ K/s for every allowed integral state. Both signs remain valid with a 0.02 K measurement error in the coordinate and an additive residual allowance of $10^{-3}$ K/s.

The calculation uses $V_A^0\le V<V_{\max}$, $0<N_B\le N_A\le N_A^0$, $h_Q\ge0$ and $e_v\in[-0.15,0.05]$ K, with $r\in[-r_-,r_+]$. The two inward margins verify~\eqref{eq:first-order-inward} throughout this region.

\begin{table}[!t]\centering\small
\caption{First-order tuning check for the reduced benchmark.}
\label{tab:reduced-screen}
\setlength{\tabcolsep}{4pt}\begin{tabularx}{\textwidth}{@{}Xll@{}}\toprule
Quantity & Value & Unit\\\midrule
$C=\Delta H/(\rho c_p)$ & $0.09$ & K m$^3$/mol\\
$k$ at 350 K & $4.7290\times10^{-8}$ & m$^3$/(mol s)\\
$t_f$ & $30$ & h\\
$T,T_{\max},e_{\mathrm{sp}}$ & $350,397,1$ & K\\
$N_A^0,c_{B,f}$ & $2.5,10869.5652$ & mol, mol/m$^3$\\
$V_A^0,V_{\max}$ & $3.2\times10^{-4},5.5\times10^{-4}$ & m$^3$\\
$K_P,K_I$ & $5.0\times10^{-7},3.0\times10^{-11}$ & m$^3$/(s K), m$^3$/(s$^2$ K)\\
$u_{\max}=r_+$ & $2.7778\times10^{-8}$ & m$^3$/s\\
$r_-$ & $8.3333\times10^{-9}$ & m$^3$/s\\
$A_s^+,A_e^-$ & $0,1.6667\times10^{-8}$ & m$^3$/s\\
$\sup(d_0/q)$ at $e_v=\eta_e$ & $5.8269\times10^{-9}$ & m$^3$/s\\
Lower edge inward margin & $1.1191\times10^{-3}$ & K/s\\
Upper edge PI margin & $1.8374\times10^{-2}$ & K/s\\
Upper edge margin with selector & $1.1094\times10^{-3}$ & K/s\\
$\kappa_Q$ & $0.2$ & s$^{-1}$\\
\bottomrule\end{tabularx}
\end{table}

\begin{figure}[!t]\centering
\includegraphics[width=\textwidth]{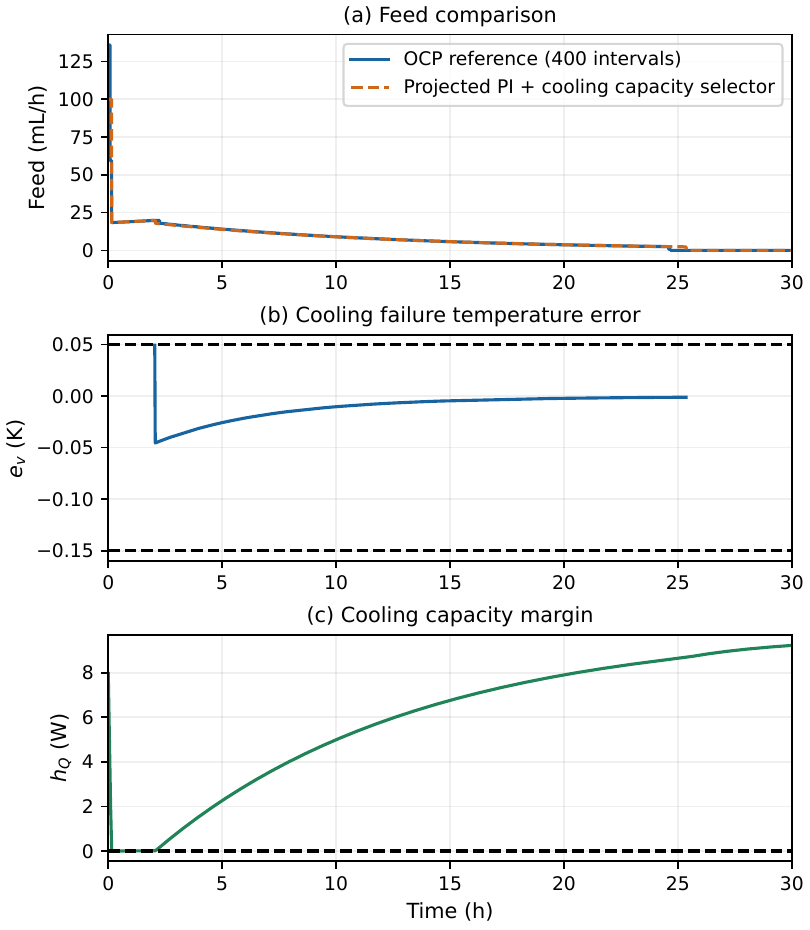}
\caption{Reduced benchmark: (a) feed; (b) cooling failure temperature error; (c) cooling capacity margin. Dashed lines in (b) bound the verified error band.}
\label{fig:benchmark_verification}
\end{figure}

Figure~\ref{fig:benchmark_verification}(a) shows that PI with the selector closely follows the OCP feed profile. Terminal conversion is 0.9198908 for PI and 0.9223016 for the OCP, a difference of 0.2411 percentage points. In panel (b), the error remains within the selected band during the interval 2.0408--25.3611 h, with $e_v\in[-0.04552,0.04939]$ K. The regional tuning condition applies until the reaction reaches the A-limited branch or feeding ends. Over the complete trajectory, the peak $T_{\rm cf}$ is 396.0455 K, and panel (c) shows that the cooling capacity margin remains nonnegative. The controller therefore approaches the optimized feed profile while respecting both thermal constraints.

\section{Industrial Semi-Batch Polymerization Case Study}
\label{sec:industrial-case}

\subsection{Industrial Case Description}
\label{subsec:industrial-case-description}
The industrial case evaluates the complete cooling demand ARC implementation, including temperature control, feed coordination and the batch completion logic. It follows the industrial gas--liquid polymerization benchmark used in prior process safety and NMPC studies \cite{li2025benchmark,zhou2026RealtimeNMPC}; detailed model equations are collected in the supplementary material. The process is a pressurized aqueous free-radical polymerization of gaseous monomer A initiated by species B. Figure~\ref{fig:mode-transition} summarizes the recipe. Heating and pressurization establish $T_{\mathrm{sp}}=351.0$ K and $P_{\mathrm{sp}}=1.5$ MPa. During reaction, temperature must remain within $T_{\mathrm{sp}}\pm 0.7$ K and pressure below $P_{\max}=1.6$ MPa until $M_{A,\mathrm{total}}=3250$ kg is charged. Monomer feed then stops, temperature is held, and reaction and gas--liquid transfer reduce pressure to $P^{\mathrm{end}}=1.0$ MPa. The temperature cascade and economic VPC operate throughout these phases; no venting is modeled.

The two controlled variables have different roles. Reactor temperature enters the kinetic constants for initiation, propagation, transfer and termination, and therefore affects both heat release and polymer moment evolution. Maintaining its recipe stabilizes the reaction conditions used to obtain the intended product. Gas-phase monomer pressure controls the driving force for gas--liquid transfer; reducing its setpoint can reduce monomer supply and heat load when cooling is limiting. Pressure setpoint relaxation is consequently an available coordination action, whereas temperature recipe compliance remains required. The pressure upper bound and any separate grade specific constraints still apply.

\begin{figure}[!hb]
  \centering
  \includegraphics[width=0.94\columnwidth]{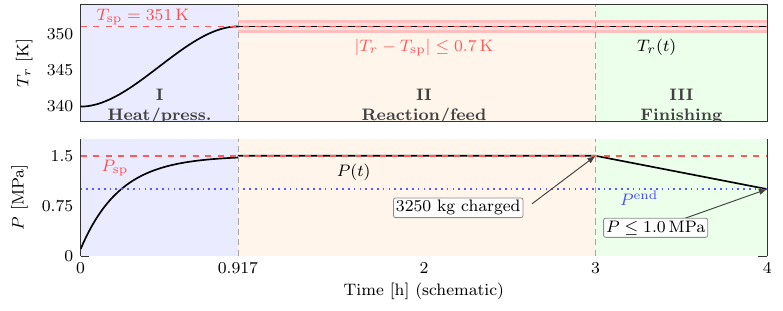}
  \caption{Temperature and pressure recipe stages. Startup ends at 3300 s; subsequent transition times are schematic.}
  \label{fig:mode-transition}
\end{figure}

The manipulated input vector is $\mathbf u_{\mathrm{in}}=[u_A,u_B]^\top$, where $u_A$ is the monomer feed to the vessel and $u_B$ is the initiator feed. Relative to the reduced model, the case adds monomer entry governed by pressure, split-range utility actuation, sensor and actuator dynamics, and an unmodeled gel effect kinetic fault.

\subsection{Industrial ARC Implementation}
\label{subsec:industrial-realization}
The feed bounds are $0\le u_A\le0.694$ kg/s and $0\le u_B\le4.0\times10^{-7}$ kg/s. The prescribed monomer charge is 3250 kg; the initiator has a total dose cap of 0.002 kg. Let $v_{\rm sr}^{\rm v}$ denote virtual split-range demand and $v_{\rm sr}$ the split-range command. Negative values request cooling; $v_{\rm sr}=-1$ is maximum cooling and $v_{\rm sr}=1$ maximum heating. The virtual split-range demand is evaluated from the cascade output before the master feasibility limit and final utility saturation. Therefore $v_{\rm sr}^{\rm v}<-1$ retains the overload magnitude even though $v_{\rm sr}\ge-1$.

On the cooling branch, $F_{\mathrm{cw}}^{\mathrm{v}}=-F_{\mathrm{cw}}^{\max}v_{\rm sr}^{\rm v}$, $F_{\mathrm{cw,sp}}=-F_{\mathrm{cw}}^{\max}v_{\rm sr,sp}$, and $F_{\mathrm{cw}}=F_{\mathrm{cw}}^{\max}\max(0,-v_{\rm sr})$ before actuator dynamics. Thus
\[
 e_{\rm VPC}=v_{\rm sr}^{\rm v}-v_{\rm sr,sp}
             =e_v/F_{\mathrm{cw}}^{\max}.
\]
Both errors are positive when cooling demand is below its setpoint. The subtraction order accounts for the negative split-range convention. The implementation uses a numerical demand guard $[-11,1]$ and a distinct physical range $[-1,1]$.

\begin{table}[!t]\centering\small
\caption{Signal definitions, signs and scaling for the industrial implementation.}
\label{tab:theory_impl_map}
\begin{tabularx}{\textwidth}{p{0.23\textwidth}p{0.20\textwidth}X}\toprule
Signal & Unit/range & Meaning\\\midrule
$F_{\mathrm{cw}}^{\mathrm{v}}$ & kg/s & Virtual cooling demand; may exceed cooling capacity\\
$v_{\rm sr}^{\rm v}$ & dimensionless; guard $[-11,1]$ & Virtual split-range demand; more negative means more cooling\\
$v_{\rm sr}$ & dimensionless; $[-1,1]$ & Split-range command before actuator dynamics\\
$v_{\rm sr,sp}$ & dimensionless; $-0.95$ & Cooling demand setpoint, leaving 5\% unused capacity\\
$e_{\rm VPC}$ & dimensionless & $v_{\rm sr}^{\rm v}-v_{\rm sr,sp}$; positive when demand is below its setpoint, negative when above\\
$\alpha_{\rm eco}$ & dimensionless; $[-40,10]$ & Economic command; positive increases initiator feed, negative lowers pressure setpoint\\
$s_{\mathrm{cool}}$ & dimensionless & $1+v_{\rm sr}^{\rm v}$ on the cooling branch; negative means demanded cooling exceeds capacity\\
$\phi$ & W & Cooling capacity residual on the tracking manifold; distinct from virtual cooling demand\\
\bottomrule\end{tabularx}\end{table}

The industrial implementation coordinates two physical channels, $u_B$ and $P_{\mathrm{sp}}$, through the signed economic command $\alpha_{\rm eco}$ (Figure~\ref{fig:industrial-control-structure}; Table~\ref{tab:theory_impl_map}). Positive $\alpha_{\rm eco}$ increases the initiator feed request, whereas negative $\alpha_{\rm eco}$ lowers $P_{\mathrm{sp}}$ when the temperature cascade is enabled; the two static maps thus activate different feed channels according to the command sign. This follows the channel physics: $u_B$ is fast but limited, whereas the pressure setpoint channel is slower and stronger because it changes gas inventory and gas--liquid transfer. The temperature cascade uses back-calculation anti-windup, and the economic PI controller clamps the shared integral state to $[-r_-,r_+]=[-1.5,1]$.

\begin{figure}[!t]
\centering
\includegraphics[width=\columnwidth]{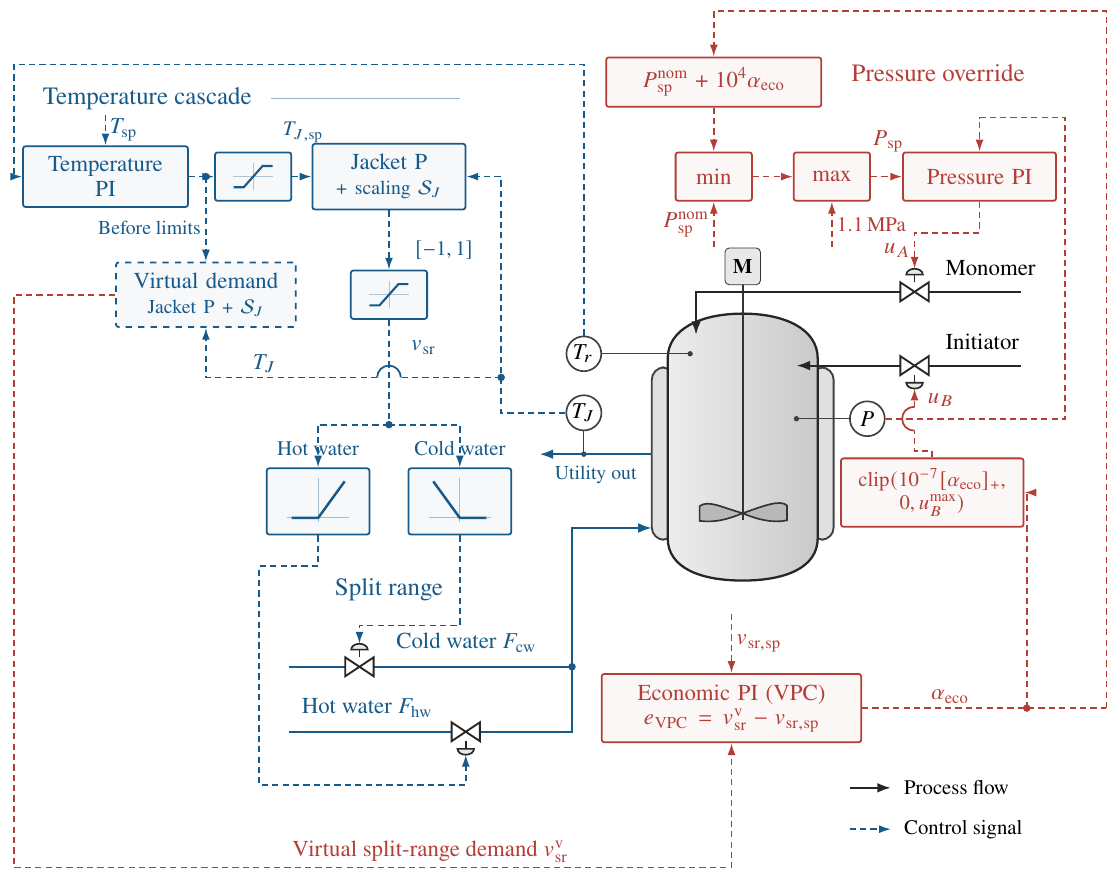}
\caption{ARC structure for the industrial reactor.}
\label{fig:industrial-control-structure}
\end{figure}

Table~\ref{tab:arc_params} gives the controller settings. Pressure error is scaled by $10^5$ Pa, giving a gain of 4 kg/(s bar). The reactor temperature master gain is in K/K; the jacket temperature slave gain precedes the nonlinear utility scaling. Static channel gains are in kg/s and Pa per economic command unit.

\begin{table}[!t]
\centering
\caption{Industrial ARC controller settings.}
\label{tab:arc_params}
\footnotesize
\setlength{\tabcolsep}{3.2pt}
\begin{tabularx}{\columnwidth}{@{}>{\raggedright\arraybackslash}X>{\raggedright\arraybackslash}Xccc@{}}
\toprule
\textbf{Loop} & \textbf{CV $\to$ MV} & \textbf{Type} & \textbf{$K_c$} & \textbf{$T_I$ (s)} \\ 
\midrule
Pressure & $P \to u_A$ & PI & 4.0 & 20 \\
$T_r$ master & $T_r \to T_{J,sp}$ & PI & 35.0 & 300 \\
$T_J$ slave & $T_J \to v_{\rm sr}$ & P & 2.0 & -- \\
Economic VPC & $v_{\rm sr}^{\rm v} \to \alpha_{\rm eco}$ & PI & 8.0 & 40 \\
Initiator map & positive $\alpha_{\rm eco} \to u_B$ & Static & $1.0{\times}10^{-7}$ & -- \\
Pressure map & negative $\alpha_{\rm eco} \to P_{\mathrm{sp}}$ & Static & $1.0{\times}10^{4}$ & -- \\
\bottomrule
\end{tabularx}
\end{table}

\FloatBarrier
Let \(v_k:=\alpha_{{\rm eco},k}\), \([v]_+ := \max(v,0)\), and \([v]_- := \min(v,0)\). With \(P_{\rm sp}^{\rm nom}=1.5~\mathrm{MPa}\) and the VPC output clipped to \([-40,10]\), the implemented maps are
\begin{align}
    u_{B,k}^{\rm fb} &= \operatorname{clip}\!\left(10^{-7}[v_k]_+,\ 0,\ u_B^{\max}\right), \label{eq:industrial-ub-map}\\
    P_{{\rm sp},k} &=
    \begin{cases}
    \operatorname{clip}\!\left(P_{\rm sp}^{\rm nom}+10^4[v_k]_-,\ 1.1~\mathrm{MPa},\ P_{\rm sp}^{\rm nom}\right), & \text{cascade enabled},\\
    P_{\rm sp}^{\rm nom}, & \text{otherwise}.
    \end{cases}\label{eq:industrial-psp-map}
\end{align}
With the settings in Table~\ref{tab:arc_params}, $e_{\rm VPC}\le-1/8$ makes the requested initiator feed zero for every permitted integral state. Actual feed still follows its 5 s actuator lag, and previously accumulated reactants can continue generating heat. The constants \(10^{-7}\) and \(10^4\) have units kg~s\(^{-1}\) and Pa per command unit, respectively.

\subsection{Implementation and tuning for the industrial case}
\label{subsec:commissioning}
We retain the controller settings in Table~\ref{tab:arc_params}. As a check on temperature control, we calculate the jacket loop response at two representative reactor states. Linearizing the jacket model with the other states held fixed gives closed-loop response times of 49.3 and 39.0 s. The supplement gives the calculation.

The startup temperature setpoint follows a quintic S curve over 3300 s to the 351.0 K hold value. Hyperbolic tangent functions smooth the gain changes in the split-range controller.

The controllers update every second using filtered reactor temperature, jacket temperature and pressure measurements. The measurement filter has a 3 s time constant. The economic PI uses the virtual split-range demand to calculate its output, which the maps above convert to initiator feed and pressure setpoint commands. The pressure controller uses a changed setpoint on the next sample, and the physical inputs follow the specified actuator lags and rate limits.

Metered shutoff stops each feed at its charge limit and holds the economic integral state. Pressure coordination remains available after the initiator cap, and temperature control continues through finishing. Online execution consists of these PI updates, static maps and limits; it uses no reactor prediction model, state estimator or optimization solver.

\paragraph{Industrial response bound check}\label{par:industrial-tuning}
The economic PI can accumulate a positive integral state while cooling capacity is unused. To remove the initiator feed request immediately when cooling is overloaded, proportional action must be strong enough to overcome this state. The upper integral limit therefore determines a sufficient proportional gain. Continued negative error also lowers the pressure setpoint. We use the economic PI control law, integral limits and actuator response to calculate how quickly these actions occur and how quickly feed can recover. Applying the local thermal safety conditions would additionally require a verified model of the reaction, cascade and filtering dynamics that produce the cooling demand.

With integral limits $[-1.5,1]$, a 1 s sample interval and cooling demand setpoint $-0.95$, the overload check assumes sustained virtual split-range demand $\le-1.10$. Gain $K\ge6.667$ is sufficient to set the initiator feed request to zero immediately. At $K=8$, a pressure setpoint reduction of 10 kPa within 30 s requires $T_I\le43.5$ s. The retained $(K,T_I)=(8,40~\mathrm{s})$ gives a pressure setpoint bound of 28 s, including use of the changed setpoint at the next sample. This is a bound on the commanded setpoint; vessel pressure has separate dynamics.

Recovery is checked under sustained virtual split-range demand $\ge-0.70$, which indicates spare cooling capacity. The retained setting restores actual initiator feed to half its maximum within 35 s, including actuator response. These calculations require active feedback and integral updates, with sufficient dose remaining for feed recovery.

The economic PI setting $(K,T_I)=(8,40~\mathrm{s})$ was already used in the ARC controller. Checking this setting against the stated response requirements illustrates the parameter calculation in Step 3 of Section~\ref{subsec:design-implications}. For a new reactor, the required response times would be chosen from its thermal margin, measurement variability and feed actuator limits.

Four settings were checked under nominal conditions (N), adverse thermal mismatch (PM$_+$) and adverse mismatch with a gel effect fault (F), using independent design seed 301. The retained $(8,40)$ meets both the response requirements and the batch criteria in all three scenarios. Setting $(8,80)$ meets the batch criteria but misses the deadlines; $(10,40)$ meets the industrial response bound check but trips in F, and $(6,40)$ exceeds the temperature recipe in F. The supplement provides the envelope derivation and all twelve outcomes.

\subsection{Evaluation protocol and metrics}
\label{subsec:evaluation-revision}

We compare ARC with nominal output feedback NMPC and parameter adaptive output feedback NMPC (A-NMPC). Numerical figures and tables label the nominal comparator NMPC. Both NMPC variants use the same reduced NMPC model, which has twelve states. They also share the objective, 30 s update and 30 min horizon. A-NMPC uses an augmented extended Kalman filter (EKF) to estimate the process states, heat of reaction $\Delta H$ and reactor--jacket thermal conductance $UA$~\cite{Pannocchia2015Offset}. The estimates of $\Delta H$ and $UA$ are initialized at their nominal values. Each second, measured reactor temperature, jacket temperature and pressure update the state and parameter estimates. The latest parameter estimates remain fixed within each NMPC prediction horizon. In the fault scenario, the predictor and EKF retain the same reaction model, without the gel kinetics or gel state; adaptation uses only the measurements and applied inputs. The supplement specifies estimator verification, optimization and common settings for recovery from constraint violations.

The common full plant model has nominal $\Delta H=327240$ J/mol and $UA=21496$ W/K. N uses nominal parameters; PM$_-$ scales $(\Delta H,UA)$ by $(0.8,1.2)$; PM$_+$ uses $(1.2,0.8)$; F adds the unmodeled gel effect to PM$_+$. Five paired seeds (200--204) give 60 main runs after settings are frozen on independent design data. All controllers receive the same saved Gaussian noise streams at one-second intervals, with standard deviations $[0.1/3,0.1/3,5000/3]$ in K, K and Pa. Monomer, initiator and utility actuators have time constants $[5,5,3]$ s and rate limits $[0.0347,2\times10^{-8},0.05]$ per second in their input units. Exact metering stops feeds at the dose caps. The reduced NMPC predictor omits the short actuator lags, while its EKF uses applied inputs.

Physical completion requires 3250 kg monomer within 1 g and $P\le1.0$ MPa without a 373.15 K trip; incomplete runs stop at 30000 s. A completion meeting the specifications additionally satisfies the $351.0\pm0.7$ K hold recipe from 3300 s and the 351.7 K/1.6 MPa upper bounds throughout the batch at every recorded sample. Table~\ref{tab:comparison-revision} reports both physical completions (Complete) and completions meeting the specifications (Spec.); batch time statistics use the latter. Temperature has priority over pressure tracking because pressure setpoint relaxation is an available feed coordination action.

Tracking metrics use true process values sampled at the end of each one-second interval, including incomplete runs. Temperature error is relative to 351.0 K from 3300 s onward; pressure deviation is relative to 1.5 MPa from 3300 s until monomer charge completion. Table~\ref{tab:tracking-revision} gives the number of runs with data in each interval, the error statistics and the fractions of samples outside the temperature limits. Computation is synchronous with zero imposed delay; timing details are in the supplement.

\subsection{Comparison with nominal NMPC and A-NMPC}
\label{subsec:comparison-revision}
ARC and A-NMPC each complete all twenty main batches within the temperature recipe band and operating upper bounds (Table~\ref{tab:comparison-revision}). A-NMPC's mean completion times in each scenario are 0.7--1.4\% shorter. ARC therefore achieves comparable batch times through regulatory feedback, while adaptation improves the NMPC baseline: nominal NMPC gives six conforming completions and trips in all ten adverse cases.

ARC has smaller worst temperature errors during holding in N and both mismatch cases (Table~\ref{tab:tracking-revision}); in F, A-NMPC reduces the worst error from 0.696 to 0.304 K, with similar mean errors (0.043/0.044 K for ARC/A-NMPC). Pressure deviations include intentional relaxation to reduce heat load (Figures~\ref{fig:industrial_pmplus} and~\ref{fig:industrial_closed_loop}). Endpoint polymer properties, EKF accuracy and computational diagnostics are reported in the supplement.

\begin{table}[!t]\centering\footnotesize
\caption{Batch completion and peak temperature over five paired seeds. Time reports mean [minimum, maximum] for specification-conforming batches; peak temperature includes all runs.}\label{tab:comparison-revision}
\setlength{\tabcolsep}{3pt}
\begin{tabular}{llcclc}\toprule
Scen. & Control & Complete & Spec. & $t_f$ (h) & Peak $T$ (K)\\\midrule
N & ARC & 5/5 & 5/5 & 3.731 [3.731,3.731] & 351.167\\
N & NMPC & 5/5 & 5/5 & 3.712 [3.701,3.721] & 351.332\\
N & A-NMPC & 5/5 & 5/5 & 3.706 [3.700,3.721] & 351.280\\
PM$_-$ & ARC & 5/5 & 5/5 & 3.712 [3.712,3.712] & 351.112\\
PM$_-$ & NMPC & 5/5 & 1/5 & 3.847 [3.847,3.847] & 351.157\\
PM$_-$ & A-NMPC & 5/5 & 5/5 & 3.682 [3.672,3.692] & 351.484\\
PM$_+$ & ARC & 5/5 & 5/5 & 4.276 [4.275,4.278] & 351.199\\
PM$_+$ & NMPC & 0/5 & 0/5 & -- & 373.350\\
PM$_+$ & A-NMPC & 5/5 & 5/5 & 4.233 [4.225,4.242] & 351.358\\
F & ARC & 5/5 & 5/5 & 4.145 [4.145,4.146] & 351.233\\
F & NMPC & 0/5 & 0/5 & -- & 373.415\\
F & A-NMPC & 5/5 & 5/5 & 4.086 [4.081,4.097] & 351.411\\
\bottomrule\end{tabular}

\end{table}

\begin{table}[!t]\centering\footnotesize
\caption{Temperature tracking and pressure coordination over five paired runs.}\label{tab:tracking-revision}
\setlength{\tabcolsep}{3pt}
\begin{tabular}{llccccc}\toprule
Scen. & Control & Runs with data & $T$ error & $P$ deviation & Outside band & Above limit\\
 & & $(T/P)$ & (K) & (kPa) & (\% of hold) & (\% of batch)\\\midrule
N & ARC & 5/5 & 0.041 / 0.153 & 0.3 / 1.6 & 0.00 & 0.000\\
N & NMPC & 5/5 & 0.022 / 0.332 & 8.5 / 52.2 & 0.00 & 0.000\\
N & A-NMPC & 5/5 & 0.021 / 0.202 & 8.3 / 50.5 & 0.00 & 0.000\\
PM$_-$ & ARC & 5/5 & 0.037 / 0.104 & 0.3 / 1.6 & 0.00 & 0.000\\
PM$_-$ & NMPC & 5/5 & 0.304 / 1.017 & 44.9 / 124.5 & 2.51 & 0.000\\
PM$_-$ & A-NMPC & 5/5 & 0.052 / 0.484 & 10.6 / 73.8 & 0.00 & 0.000\\
PM$_+$ & ARC & 5/5 & 0.026 / 0.167 & 3.2 / 44.7 & 0.00 & 0.000\\
PM$_+$ & NMPC & 0/0 & -- / -- & -- / -- & -- & 35.222\\
PM$_+$ & A-NMPC & 5/5 & 0.026 / 0.177 & 12.3 / 56.5 & 0.00 & 0.000\\
F & ARC & 5/5 & 0.043 / 0.696 & 140.7 / 199.4 & 0.00 & 0.000\\
F & NMPC & 0/0 & -- / -- & -- / -- & -- & 29.717\\
F & A-NMPC & 5/5 & 0.044 / 0.304 & 121.8 / 192.9 & 0.00 & 0.000\\
\bottomrule\end{tabular}
\par\smallskip\begin{minipage}{\textwidth}\footnotesize
Error pairs: average mean absolute error / maximum absolute error across runs. Percentages: worst per-run fractions outside $351.0\pm0.7$ K during holding and above 351.7 K over the full batch.
\end{minipage}

\end{table}

\begin{figure}[!t]\centering
\includegraphics[width=0.96\textwidth]{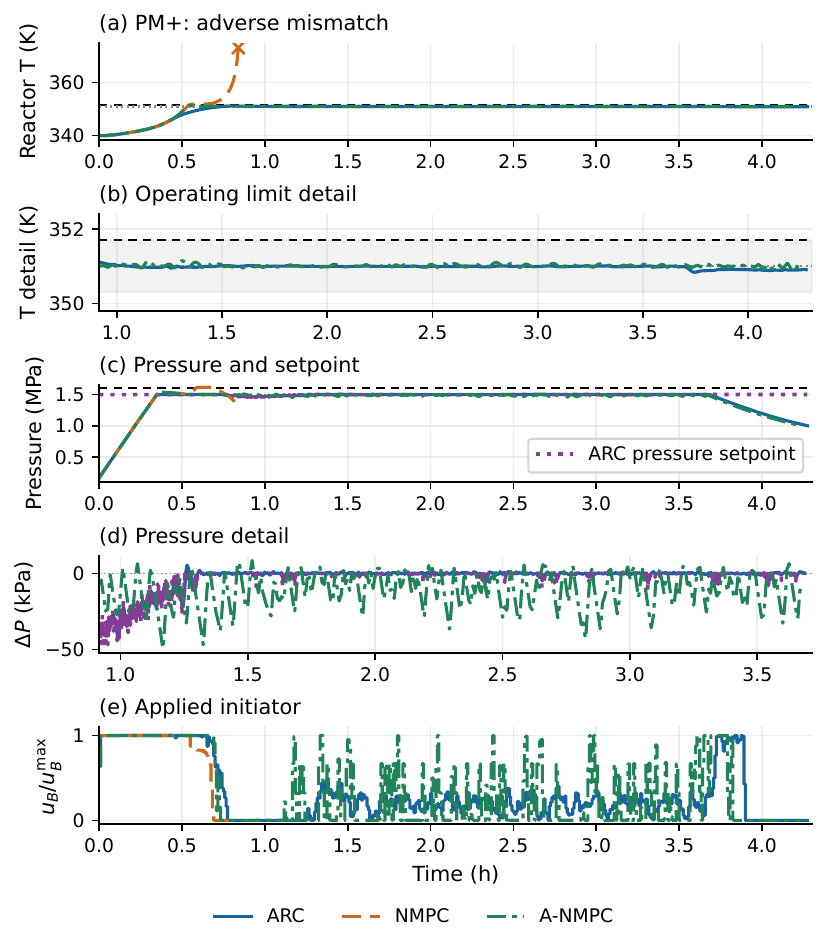}
\caption{Closed-loop response under adverse thermal mismatch (PM$_+$), paired seed 200. Black dashed lines: operating limits; purple dotted line: ARC pressure setpoint; shading: temperature recipe band. Initiator feed: 30 s block averages.}
\label{fig:industrial_pmplus}
\end{figure}

\begin{figure}[!t]\centering
\includegraphics[width=0.96\textwidth]{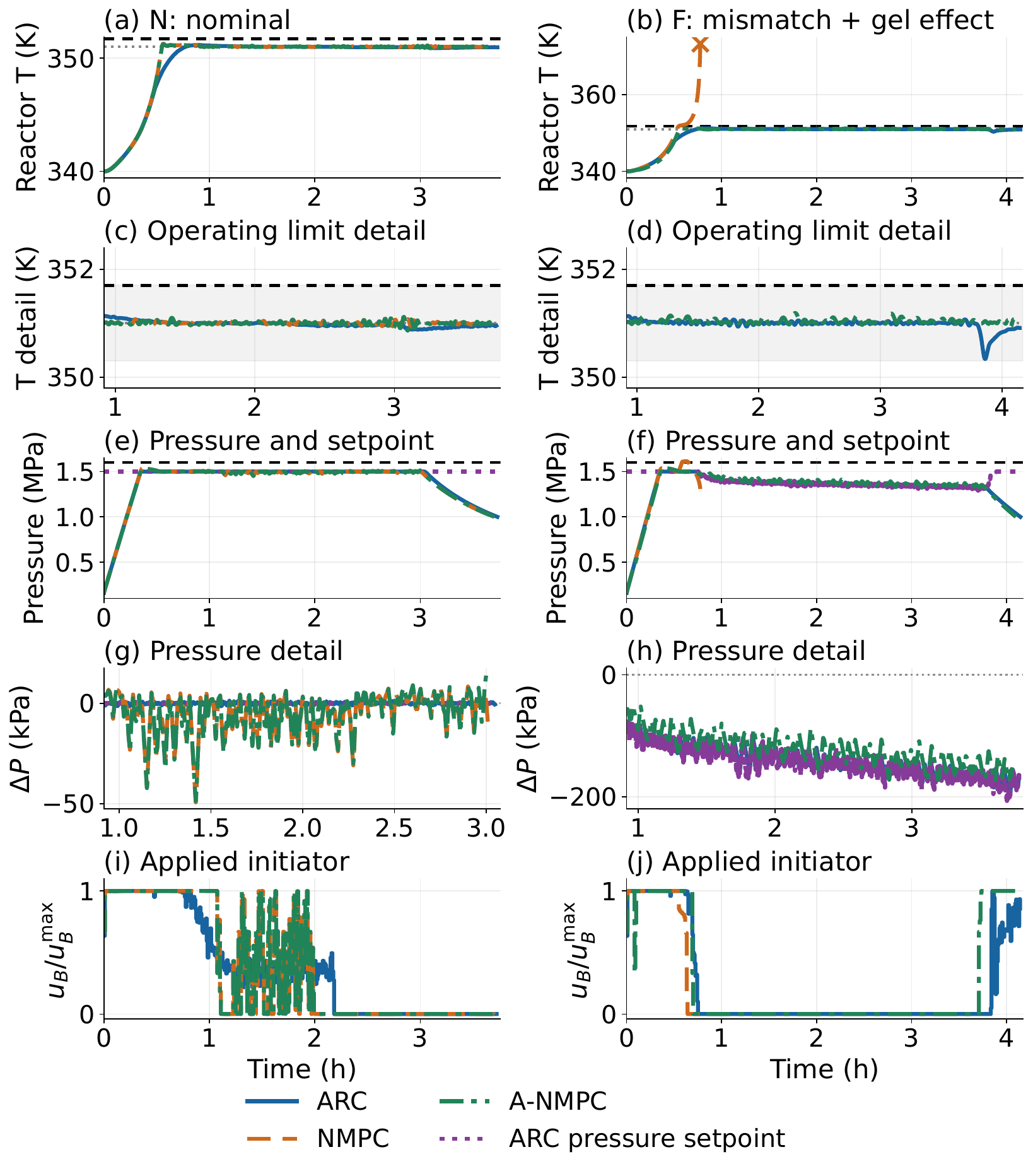}
\caption{Closed-loop responses under nominal conditions (N, left) and thermal mismatch with gel effect (F, right), paired seed 200. Plotting conventions follow Figure~\ref{fig:industrial_pmplus}.}
\label{fig:industrial_closed_loop}
\end{figure}
\FloatBarrier

\subsection{Feed channel roles and sensitivity}
\label{subsec:ablation}\label{subsec:safety-validation}
The paired seed-0 ablations isolate initiator feedback and pressure setpoint relaxation. Fixed $P_{\rm sp}$ removes the latter; recipe $u_B$ replaces initiator feedback by the benchmark recipe; recipe only removes both. Under N, initiator feedback shortens completion from 4.141 to 3.731 h. Under PM$_+$ and F, both variants with fixed pressure setpoints trip. Pressure setpoint relaxation with recipe initiator completes PM$_+$ in 4.505 h, versus 4.278 h for full ARC, but exceeds the temperature limit in F (353.387 K versus 351.227 K for full ARC). These results support the distinct roles of the two channels. The supplement reports all twelve runs and the recipe definition.

Additional paired tests triple noise standard deviations and double actuator time constants with controller and estimator settings retained. Under adverse mismatch, ARC and A-NMPC peaks rise to 352.082 and 352.015 K, above 351.7 K; A-NMPC under N also exceeds the limit (351.738 K). These results show that tuning needs to be reassessed when noise or actuator response changes. The supplement gives these tests and the separate comparison of three tuning designs.

\section{Discussion and Conclusions}
\label{sec:discussion}
\label{sec:conclusions}

\subsection{Relation to other control approaches}
\label{subsec:scope-revision}
Table~\ref{tab:approach-comparison} compares how the approaches select feedback variables, account for constraints and use process models. Economic feedback and joint operation and control design connect operating policies to feedback structures. Our design uses the active cooling constraint to select virtual cooling demand. The local safety analysis guides the PI realization with integral projection and gives gain and integral limit conditions. The designer selects the basic loops and allocates the available feed actuators, then tests their interaction with economic VPC over complete batches.

\begin{table}[!htbp]
\centering
\makebox[\textwidth][c]{%
\begin{minipage}{\dimexpr\paperwidth-48mm\relax}
\caption{Design basis and implementation requirements of representative control approaches.}
\label{tab:approach-comparison}
\begingroup
\renewcommand{\baselinestretch}{1}\small
\setlength{\tabcolsep}{3pt}
\renewcommand{\arraystretch}{1.10}
\begin{tabularx}{\textwidth}{@{}>{\raggedright\arraybackslash}p{0.20\textwidth}*{5}{>{\raggedright\arraybackslash}X}@{}}
\toprule
Design aspect
& NMPC~\cite{Rawlings2017MPCbook}
& Economic feedback~\cite{Bonvin2003ActiveConstraints,Ferreira2025Feedback}
& Joint operation and control design~\cite{Alvarez2004Joint,Andrade2025Joint}
& Nonlinear temperature control~\cite{Nguyen2023Jacket}
& Proposed ARC\textsuperscript{a} \\
\midrule
CV--MV selection basis
& Specified structure
& Optimality conditions
& Economic and dynamic analysis
& Given target and coolant input
& Economic analysis \\
\addlinespace[4pt]
Main analytical results
& Feasibility; stability
& Optimality conditions
& Motion stability
& Temperature convergence
& Local optimality; safety\textsuperscript{b} \\
\addlinespace[4pt]
Active constraints and switching
& Optimization active set
& Constraint feedback; phase switches
& Nominal paths; stop events
& Not addressed
& Virtual cooling demand; dose limits \\
\addlinespace[4pt]
Actuator saturation
& Input bounds
& Varies
& Nominal avoidance; IMC support~\cite{Alvarez2004Joint}
& Not treated
& Limits; integral projection \\
\addlinespace[4pt]
Offline design information
& Dynamic prediction model
& Model for policy design
& Dynamic process models
& Thermal models and data
& Local model and response bounds \\
\addlinespace[4pt]
Online nonlinear optimization
& $\checkmark$ & --- & --- & --- & --- \\
\addlinespace[4pt]
Online reactor prediction or state/rate estimation
& $\checkmark$ & $\triangle$ & $\checkmark$ & $\checkmark$ & --- \\
\addlinespace[4pt]
Main tuning choices
& Weights; horizon; estimator
& Feedback gains
& Controller and estimator speeds
& Gains satisfying stability conditions
& PI gains and integral bounds \\
\bottomrule
\end{tabularx}
\par\smallskip\noindent
$\checkmark$: required; ---: not required; $\triangle$: varies by implementation. IMC: internal model control. State/rate estimation refers to reactor states or reaction rates. Analytical results require their stated assumptions; NMPC guarantees refer to the cited framework.\par
\noindent\textsuperscript{a}ARC entries concern the economic VPC loop and its industrial implementation; basic loop configuration and detailed actuator allocation retain engineering choices.\par
\noindent\textsuperscript{b}Safety covers regional bounds and exact reduced model thermal feasibility with the cooling capacity selector.
\endgroup
\end{minipage}%
}
\end{table}
\FloatBarrier

\phantomsection\label{par:operating-objectives}
Recipe optimization and ARC can be combined. Recipe optimization specifies temperature profiles and endpoint requirements~\cite{Alvarez2004Joint,Ferreira2025Feedback}, while ARC adjusts feed within those specifications. When tracking necessary conditions of optimality (NCO)~\cite{Bonvin2003ActiveConstraints}, the economic VPC can implement cooling constraint feedback during the corresponding active phase. Feed adjustments must remain compatible with product requirements.

With fixed qualified output $Y$ and turnaround time $t_d$, maximizing productivity $Y/(t_f+t_d)$ is equivalent to minimizing completion time. Energy or product property objectives may change the active constraints and economic feed allocation. The local PI conditions apply wherever their input-output assumptions hold.

\subsection{Economic VPC and practical operation}
\label{subsec:economic-vpc-discussion}
Conventional VPC restores an input toward an interior value or coordinates inputs near saturation. Its controlled variable can be a controller output or a loop setpoint~\cite{Skogestad2023ARC}. Here, economic VPC uses the temperature control system's cooling request to coordinate feed near the cooling capacity limit. Retaining this request without clipping it to available capacity lets the controller respond to excess cooling demand. Integral projection bounds the stored feed request, so sufficient proportional action can set the request to zero even after prolonged operation with unused cooling capacity.

\phantomsection\label{par:arc-maintenance}
Plant operators and control engineers can adjust the PI parameters, setpoints and limits within the installed control system. For example, lowering the cooling demand setpoint increases the requested cooling reserve. Section~\ref{subsec:design-implications} describes how to configure and tune this coordination.

\subsection{Conclusions}
\phantomsection\label{par:conclusion}
We have developed an advanced regulatory control (ARC) design that coordinates feed with the cooling needed to maintain the prescribed reactor temperature. Minimum-time analysis motivates an economic valve position control (VPC) loop that adjusts feed using the temperature control system's cooling request. This virtual cooling demand is retained when the cooling actuator saturates, allowing the loop to reduce feed during cooling overload and restore feed as capacity becomes available. Local safety analysis specifies a proportional-integral controller with a bounded integral state and gives gain and integral limit conditions derived from a local input-output model.

The reduced benchmark verifies the tuning conditions and thermal feasibility with a cooling capacity selector. In the industrial simulations, ARC and parameter adaptive nonlinear model predictive control (A-NMPC) each complete all twenty main batches within the temperature recipe and operating limits at the recorded samples. A-NMPC's mean batch times are 0.7--1.4\% shorter, so ARC achieves comparable batch times without online optimization. Higher noise and slower actuators produce temperature exceedances in the stress tests. Further work will address joint design of the regulatory loops and economic VPC, accounting for measurement noise and actuator response.

\section*{Code availability}
The simulation and analysis code is available at \url{https://github.com/Daakuang/arc4batch}.

\appendix
\renewcommand{\thesection}{\Alph{section}}
\section{Proof Details for the Reduced ARC Synthesis}
\label{sec:supp-proof-details}

This appendix proves the two optimality propositions. Proofs of the local tuning results are given in Appendix~\ref{sec:supp-auxiliary-screens}.

\subsection{No Sustained Active Cooling Below the Temperature Bound}
\label{sec:supp-proof-temp-boundary}

Proposition~\ref{prop:temp-boundary} states that an optimal minimum-time batch cannot spend an interval of positive measure with $T<T_{\mathrm{ref}}$ while nonzero cooling is still applied.

The proof uses regularity and heat removal monotonicity from Assumption~\ref{ass:regularity}, the positive thermal depletion sensitivity from Assumption~\ref{ass:depletion}, and feasible local feed retiming from Assumption~\ref{ass:dose-slack}.

We proceed by contradiction.
Step 1 applies a local cooling reduction variation to exploit the available temperature margin and create a positive depletion margin inside an admissible trajectory tube. Step 2 restores the required feed dose on the shortened horizon. Step 3 uses local retiming feasibility to keep the path constraints satisfied and compares the dose adjustment with the depletion margin.

 Let $(\mathbf u_{\mathrm{in}}^\star ,F_{\mathrm{cw}}^\star ,\mathbf x^\star ,t_f^\star )$ be an optimal solution, with $\mathbf x^\star=(\mathbf n^\star,T^\star)$.
Assume there exists a measurable interval $\mathcal I\subset[0,t_f^\star ]$ of positive measure such that
$T^\star (t)<T_{\mathrm{ref}}(t)$ and $F_{\mathrm{cw}}^\star (t)>0$ on $\mathcal I$.

\noindent\emph{Step 1: lift the temperature and create a depletion margin.}
Choose a bounded nonnegative reduction profile $h_F$ from Assumption~\ref{ass:depletion}, supported on a compact subset of positive measure $\mathcal I_0\subset\mathcal I$ with uniform temperature and cooling flow margins. A measurable cooling input need not admit a subinterval with uniform flow slack; the compact subset suffices. By Assumptions~\ref{ass:regularity} and~\ref{ass:depletion}, there exists a sufficiently small $\delta>0$ such that the perturbed cooling
$\tilde F_{\mathrm{cw}}(t)=F_{\mathrm{cw}}^\star(t)-\delta h_F(t)$ on $\mathcal I_0$ (and $\tilde F_{\mathrm{cw}}=F_{\mathrm{cw}}^\star $ elsewhere) generates a trajectory in the admissible local tube, and gives a positive first variation of the integrated depletion functional.

Using the residual inventory $n_R$ from \eqref{eq:residual-reactant}, the trajectory satisfies
\[
\dot n_R(t)= -\chi(\mathbf n(t),T(t)) + w_R^\top \mathbf W_{\mathrm{in}}\mathbf u_{\mathrm{in}}(t).
\]
Thus there exists $\varepsilon_R>0$ and, by continuity of the endpoint map, an arbitrarily small $\tau>0$ such that for the same feed input $\mathbf u_{\mathrm{in}}^\star$ the perturbed trajectory satisfies
\[
\tilde n_R(t_f^\star-\tau)\le n_R^\star(t_f^\star)-\varepsilon_R .
\]

\noindent\emph{Step 2: restore the required dose on the shorter horizon.}
If the batch is terminated at $t_f^\star-\tau$ with the original feed schedule, the total charged amount is short of the required dose by
\[
\Delta U_i:=U_i^{\max}-\int_0^{t_f^\star-\tau}u_{i,\mathrm{in}}^\star(t)\,dt\ge0,\qquad i\in\mathcal D_C.
\]
For $\tau$ small enough, $\Delta U_i=O(\tau)$ for each feed component. Assumption~\ref{ass:dose-slack} gives, for each dose-constrained feed component $i$ with $\Delta U_i>0$, a set of positive measure $\mathcal J_i$ on which $u_{i,\mathrm{in}}^\star(t)\le u_i^{\max}-\delta_u$. Since the shortened batch ends at \(t_f^\star-\tau\), restrict \(\mathcal J_i\) to \(\mathcal J_i\cap[0,t_f^\star-\tau]\) and use the same symbol for this restricted set. For sufficiently small \(\tau\), it still has positive measure. To satisfy the dose constraint \eqref{eq:ct-dose} on the shorter horizon, define
\[
\hat u_{i,\mathrm{in}}(t)=u_{i,\mathrm{in}}^\star (t)+\frac{\Delta U_i}{|\mathcal J_i|}\mathbf 1_{\mathcal J_i}(t),\qquad i\in\mathcal D_C\ \text{with }\Delta U_i>0,
\]
where $\mathbf 1_{\mathcal J_i}$ is the indicator of $\mathcal J_i$. For $i\notin\mathcal D_C$ or $\Delta U_i=0$, set $\hat u_{i,\mathrm{in}}=u_{i,\mathrm{in}}^\star$ without defining $\mathcal J_i$. The indicator confines each nonzero adjustment to the selected set of times with unused feed capacity. Since $\Delta U_i=O(\tau)$, choosing $\tau$ small enough gives $\Delta U_i/|\mathcal J_i|\le \delta_u$, hence $0\le \hat{\mathbf u}_{\mathrm{in}}(t)\le \mathbf u^{\max}$ and
$\int_0^{t_f^\star-\tau}\hat u_{i,\mathrm{in}}(t)\,dt=U_i^{\max}$ for $i\in\mathcal D_C$. Streams subject only to dose caps are unchanged on the shortened horizon, so their cumulative caps remain satisfied.

\noindent\emph{Step 3: compare the dose adjustment with the depletion margin.}
By Assumption~\ref{ass:dose-slack}, the $O(\tau)$ retiming adjustment above stays in an admissible local trajectory tube and preserves the path constraints for sufficiently small $\tau$. The $C^1$ dynamics on a compact local tube imply continuous dependence of the endpoint on the feed perturbation: $|\hat n_R(t_f^\star-\tau)-\tilde n_R(t_f^\star-\tau)|\le L_U\|\hat{\mathbf u}_{\mathrm{in}}-\mathbf u_{\mathrm{in}}^\star\|_{L^1}=O(\tau)$ for a finite local constant $L_U$. This bound includes the effect of the additional feed on subsequent kinetics.
Choosing $\tau$ sufficiently small, this increment is dominated by the extra depletion margin $\varepsilon_R$, so the terminal constraint still holds:
$\hat n_R(t_f^\star -\tau)\le n_R^{\min}$.
This constructs a feasible solution with final time $t_f^\star -\tau<t_f^\star $, contradicting optimality.
Therefore Proposition~\ref{prop:temp-boundary} holds.

\medskip\noindent\emph{Checking the thermal depletion sensitivity.}
One direct test of Assumption~\ref{ass:depletion} uses the variational equation. Let $z_T=(\mathbf n,T)$, let $A_T(t)$ be the state Jacobian along the candidate trajectory, and let $b_F(t)=(\mathbf 0,-\partial_Fq_{\mathrm{ex}}/C_P)$ be its cooling input sensitivity. For a nonnegative cooling reduction profile $h_F$ supported on the cooled interior interval, solve
\begin{equation}\label{eq:thermal-sensitivity-test}
 \dot\eta_T=A_T\eta_T-b_Fh_F,\qquad\eta_T(0)=0.
\end{equation}
The strict sign $w_R^\top\eta_{T,\mathbf n}(t_f)<0$ verifies the required first-order depletion improvement for that variation. Path feasibility is checked separately over the remaining campaign.

\subsection{Active-Constraint Structure of the Time-Optimal Feed}
\label{sec:supp-proof-boundary-tv}

Proposition~\ref{prop:boundary-tv} states that, on the temperature tracking manifold and for almost every time at which no other state constraint dominates, the pointwise Hamiltonian feed minimization is a linear program over a box intersected with one cooling capacity half-space. The optimizer is determined by the input box when cooling margin remains and must satisfy $\phi=0$ when cooling capacity is limiting.

The proof uses Assumption~\ref{ass:regularity}, the affine dependence of $\phi$ on feed, and the pointwise Hamiltonian minimization property.

Step 1 identifies the feasible polytope on the reduced tracking manifold. Step 2 reads the cooling margin case from the input box. Step 3 translates the cooling-limited half-space into maximum cooling demand. Singular arcs and additional active state constraints are outside this local statement.

\noindent\emph{Step 1: reduce the Hamiltonian minimization to a linear program.}
The pointwise minimization in \eqref{eq:ham-min-tv} is over
\[
\mathcal U_T=
\{\mathbf u\mid
\mathbf 0\le \mathbf u\le \mathbf u^{\max},\ 
\phi(\mathbf n,\mathbf u,t)\le0\}.
\]

\noindent\emph{Step 2: read the cooling margin case from the input box.}
When the optimizer lies strictly inside the cooling capacity half-space, the remaining feasible set is locally only the input box. Minimizing the linear form $\boldsymbol{\sigma}^\top\mathbf u$ over the box gives the displayed componentwise bang-bang rule, with the usual singular case when $\sigma_i=0$.

\noindent\emph{Step 3: translate the cooling-limited half-space into cooling capacity.}
When cooling capacity is limiting, the defining equality is $\phi=0$. By \eqref{eq:phi-tv}, this equality is $q_{\mathrm{req}}=q_{\mathrm{ex}}(T_{\mathrm{ref}},F_{\mathrm{cw}}^{\max})$, which is equivalent to $F_{\mathrm{cw}}^{\mathrm{th}}=F_{\mathrm{cw}}^{\max}$ by the inverse definition \eqref{eq:Freq-tv}. This proves Proposition~\ref{prop:boundary-tv}.

\section{Auxiliary Conditions for Endpoint Interpretation}
\label{sec:supp-auxiliary-screens}

\subsection{Proof of Theorem~\ref{thm:band-invariance}}\label{sec:proof-local-endpoint}
Substitute the boundary value, slope bound and corresponding budget into~\eqref{eq:window-lower} or~\eqref{eq:window-upper}. The two stated inequalities give the conclusions. \hfill$\square$

\subsection{Proof of Corollary~\ref{cor:first-order-band}}\label{sec:proof-first-order-inward}
At the lower edge, projection gives $\alpha\le A_s^+$, so $\dot e_v\ge0$ by~\eqref{eq:first-order-inward}. At the upper edge, $\alpha\ge A_e^-$ gives $\dot e_v\le0$. The interval boundary is therefore inward pointing for all allowed integral states and residuals. The scalar barrier argument, applied until region exit, proves forward invariance. \hfill$\square$

\subsection{Projected PI feed withdrawal}\label{sec:supp-pi-curvature}
Projection preserves $r\in[-r_-,r_+]$: at the upper boundary $\dot r=\min(K_Ie_v,0)\le0$, and at the lower boundary $\dot r=\max(K_Ie_v,0)\ge0$. Consequently, $e_v\le-(r_++\alpha_b)/K_P$ implies $\alpha_b+K_Pe_v+r\le0$ and hence zero feed request for every permitted integral state. Actual feed withdrawal additionally depends on actuator response.

\subsection{Between-sample containment}\label{sec:supp-proof-continuous-band}
Let $|\dot e_v|\le L_e$ on windows of length $\tau$. A trajectory starting in $[-\hat\eta_s,\hat\eta_e]$ remains in $[-\eta_s,\eta_e]$ during that window if
\[
 \hat\eta_s+L_e\tau\le\eta_s,\qquad \hat\eta_e+L_e\tau\le\eta_e,
\]
because $|e_v(t)-e_v(t_k)|\le L_e\tau$. Repeated containment requires membership of the inner band at every window start; the endpoint test alone does not establish it.

\subsection{Thermal interpretation}\label{sec:supp-continuous-bridge}
A verified nonincreasing envelope $\phi(t)\le\Psi(e_v(t))$ on region $\mathcal R$ converts $e_v\ge\eta_v^\star$ into $\phi<0$ when $\Psi(\eta_v^\star)<0$. Under ideal tracking and static actuation, a candidate is
\[
 \Psi_{\rm mod}(z)=\sup_{t\in\mathcal R}\left[
 q_{\rm ex}\!\left(T_{\rm ref}(t),\operatorname{sat}(F_{\rm cw,sp}-z;0,F_{\rm cw}^{\max})\right)
 -q_{\rm ex}(T_{\rm ref}(t),F_{\rm cw}^{\max})\right].
\]
Monotonicity of heat removal and saturation makes this envelope nonincreasing. Away from exact tracking, the envelope must also cover cascade error and unmodeled dynamics.

\begingroup
\small
\setlength{\bibsep}{0pt}
\bibliographystyle{elsarticle-num}
\bibliography{ref}
\endgroup

\clearpage
\makeatletter
\@removefromreset{equation}{section}
\makeatother
\setcounter{section}{0}
\renewcommand{\thesection}{S\arabic{section}}
\setcounter{table}{0}
\renewcommand{\thetable}{S\arabic{table}}
\setcounter{figure}{0}
\renewcommand{\thefigure}{S\arabic{figure}}
\setcounter{equation}{0}
\renewcommand{\theequation}{S\arabic{equation}}
\setcounter{theorem}{0}
\renewcommand{\thetheorem}{S\arabic{theorem}}
\renewcommand{\theHsection}{S\arabic{section}}
\renewcommand{\theHtable}{S\arabic{table}}
\renewcommand{\theHfigure}{S\arabic{figure}}
\renewcommand{\theHequation}{S\arabic{equation}}
\renewcommand{\theHtheorem}{S\arabic{theorem}}
\renewcommand{\theHsubsection}{S\arabic{section}.\arabic{subsection}}
\renewcommand{\theHsubsubsection}{S\arabic{section}.\arabic{subsection}.\arabic{subsubsection}}

\begin{center}
{\Large\bfseries Supplementary Material\par}
\vspace{0.55em}
{\large\bfseries ``A Theory-Guided Advanced Regulatory Control Synthesis for\\
Cooling-Limited Exothermic Semi-Batch Reactors''\par}
\vspace{0.85em}
Chenchen Zhou\textsuperscript{a}\quad Jose Matias\textsuperscript{a}\par
\vspace{0.35em}
{\small \textsuperscript{a}Chemical and Biochemical Reactor Engineering and Safety (CREaS), KU Leuven,\\
Jan Pieter de Nayerlaan 5, Sint-Katelijne-Waver, 2860, Belgium\par}
\end{center}
\vspace{1.0em}

\section{Reduced-Benchmark Reproduction and Numerical Checks}
\label{supp-sec:supp-reduced-benchmark}
The main manuscript gives the complete first-order model, controller gains and boundary calculations. The coordinate $e_v=T_{\max}-T_{\rm cf}-1$ K satisfies $\dot e_v=Ck(N_A/V)c_B-C(c_{B,f}-c_B)u/V$ on the B-limited branch. A second differentiation introduces $\dot u$; the local relation between feed and this coordinate is first order.

We solve the optimal control problem (OCP) on grids of 100 and 400 intervals, with feed held constant within each interval. Each optimized feed profile is then applied to the reactor model using an independent numerical integrator to check the trajectory and thermal constraints. For the PI controller with the cooling capacity selector, we also check the coordinate derivative and refine the closed-loop integration. The physical parameters are $\Delta H=162000$ J/mol, $\rho=900$ kg/m$^3$, $c_p=2000$ J/(kg K), $k_0=9.72\times10^4$, $E_a=82500$ J/mol, $R=8.314$ J/(mol K), $T=350$ K, and $t_f=30$ h. Cooling uses $U=20$ W/(m$^2$ K), geometry parameter $\sigma=0.076$ m, coolant minimum 335 K, and area $2V/\sigma+\pi\sigma^2$.

The 400-interval reference has conversion 0.9223016; independent integration agrees within $1.2\times10^{-11}$ in conversion. Its maximum cooling failure temperature is 396.9990 K and maximum cooling capacity residual is $-7.0740\times10^{-4}$ W. The projected PI with the cooling capacity selector gives conversion 0.9198908 and its maximum cooling failure temperature is 396.0455 K. The exact first-order derivative and an independent directional difference agree within $2.02\times10^{-10}$ K/s on the checked regional samples. These numerical tests support the regional analytic calculation. The error band guarantee applies while the trajectory remains in the stated operating region.

\section{Detailed Reactor Modeling}
\label{supp-sec:detailed_model}
The industrial polymerization model combines polymer chain moments with mass and energy balances.

\subsection{Reaction Kinetic Scheme}
The kinetic scheme models the free-radical polymerization of monomer A initiated by species B, in the presence of a chain transfer agent C. The set of elementary reactions considered is listed in Table~\ref{supp-table:kinetics}.

\begin{table}[htbp]
	\caption{Free-radical polymerization kinetic scheme.}
	\label{supp-table:kinetics}
	\centering
	\begin{tabular}{l|ll}
		\toprule
		\textbf{Mechanism Step} & \textbf{Reaction Stoichiometry} & \textbf{Rate Constant} \\
		\midrule
		Initiator Decomposition & $\rm B \xrightarrow{\textit{k}_{1}} 2R^*$ & $k_{1}=k_{10}\exp\left({-{E_{1}}/({\rm R}T_{\rm r})}\right)$\\
		Chain Initiation & $\rm A + R^* \xrightarrow{\textit{k}_2} P_1$ &$k_{2}=k_{20}\exp\left({-{E_{2}}/({\rm R}T_{\rm r})}\right)$\\
		Propagation & $\rm P_n + A \xrightarrow{\textit{k}_3} P_{n+1}$ &$k_{3}=k_{30}\exp\left({-{E_{3}}/({\rm R}T_{\rm r})}\right)$\\
		Transfer to Monomer & $\rm P_n + A \xrightarrow{\textit{k}_{4}} D_n + P_1$ &$k_{4}=k_{40}\exp\left({-{E_{4}}/({\rm R}T_{\rm r})}\right)$ \\
		Transfer to Agent C & $\rm P_n + C \xrightarrow{\textit{k}_{5}} D_n + R^*$ &$k_{5}=k_{50}\exp\left({-{E_{5}}/({\rm R}T_{\rm r})}\right)$\\
		Termination & $\rm P_n + P_m \xrightarrow{\textit{k}_7} D_{n+m}$ &$k_{7}=k_{70}\exp\left({-{E_{7}}/({\rm R}T_{\rm r})}\right)$\\
		\bottomrule
	\end{tabular}
\end{table}

\subsection{Moment-Based Dynamic Model}
The system dynamics are described by the leading moments of the live ($\lambda_k$) and dead ($\mu_k$) polymer chain distributions, coupled with the species concentrations and reactor energy balance. Liquid volume $V_l$ is in L, and concentrations and concentration moments are in mol/L. Multiplying a concentration or moment by $V_l$ gives the corresponding inventory $N_i$ or moment total $N_{\lambda j}$ or $N_{\mu j}$.

\subsubsection{Polymer Moments and Species Balances}
The evolution of the zeroth, first, and second moments for live ($\lambda$) and dead ($\mu$) chains is governed by:
\begin{align}
    \frac{d(V_{\rm l} \lambda_0)}{dt} &= \left( k_2 c_{\rm A} c_{\rm R} - (k_5 c_{\rm C}) \lambda_0 - k_7 \lambda_0^2 \right) V_{\rm l} \label{supp-eq:lam0} \\
    \frac{d(V_{\rm l} \lambda_1)}{dt} &= \left( k_2 c_{\rm A} c_{\rm R} + (k_3+k_4)c_{\rm A} \lambda_0 - (k_4 c_{\rm A} + k_5 c_{\rm C})\lambda_1 - k_7 \lambda_0 \lambda_1 \right) V_{\rm l} \label{supp-eq:lam1} \\
    \frac{d(V_{\rm l} \lambda_2)}{dt} &= \left( k_2 c_{\rm A} c_{\rm R} + (k_3+k_4)c_{\rm A} \lambda_0 + 2k_3 c_{\rm A} \lambda_1 - (k_4 c_{\rm A} + k_5 c_{\rm C})\lambda_2 - k_7 \lambda_0 \lambda_2 \right) V_{\rm l} \label{supp-eq:lam2} \\
    \frac{d(V_{\rm l} \mu_0)}{dt} &= \left( (k_4 c_{\rm A} + k_5 c_{\rm C})\lambda_0 + \frac{1}{2}k_7 \lambda_0^2 \right) V_{\rm l} \label{supp-eq:mu0} \\
    \frac{d(V_{\rm l} \mu_1)}{dt} &= \left( (k_4 c_{\rm A} + k_5 c_{\rm C})\lambda_1 + \frac{1}{2}k_7 (\lambda_0 \lambda_1 + \lambda_1 \lambda_0) \right) V_{\rm l} \label{supp-eq:mu1} \\
    \frac{d(V_{\rm l} \mu_2)}{dt} &= \left( (k_4 c_{\rm A} + k_5 c_{\rm C})\lambda_2 + \frac{1}{2}k_7 (\lambda_0 \lambda_2 + 2\lambda_1^2 + \lambda_2 \lambda_0) \right) V_{\rm l} \label{supp-eq:mu2}
\end{align}

The concentrations of small molecules ($R^*, A, B, C$) follow:
\begin{align}
    \frac{d(V_{\rm l} c_{\rm R})}{dt} &= \left( 2f k_1 c_{\rm B} - k_2 c_{\rm A} c_{\rm R} + (k_5 c_{\rm C}) \lambda_0 \right) V_{\rm l} \label{supp-eq:cR} \\
    \frac{d(V_{\rm l} c_{\rm A})}{dt} &= 1000 \frac{F_{\rm g2l}}{M_{\rm gA}} - \left( k_2 c_{\rm R} + (k_3+k_4)\lambda_0 \right) c_{\rm A} V_{\rm l} \label{supp-eq:cA} \\
    \frac{d(V_{\rm l} c_{\rm B})}{dt} &= 1000 \frac{F_{\rm B}}{M_{\rm gB}} - 2f k_1 c_{\rm B} V_{\rm l} \label{supp-eq:cB} \\
    \frac{d(V_{\rm l} c_{\rm C})}{dt} &= - k_5 \lambda_0 c_{\rm C} V_{\rm l} \label{supp-eq:cC}
\end{align}
Here $f$ is the initiator efficiency, and $F_{g2l}$ is the mass transfer rate of monomer A from the gas phase to the liquid phase.

\subsubsection{Thermal and Phase Equilibrium}
The reactor temperature $T_{\rm r}$, jacket temperature $T_{\rm J}$, and pressure $P$ dynamics incorporate the heat of reaction, heat exchange with the jacket, and gas-liquid mass transfer:
\begin{align}
    C_{\rm eff} \frac{dT_{\rm r}}{dt} &= F_{\rm A} C_{p\rm A}(T_{\rm in} - T_{\rm r}) + UA(T_{\rm J} - T_{\rm r}) + (-\Delta H_{\rm rxn}) R_{\rm poly} V_{\rm l} - Q_{\rm loss} + Q_{\rm stir} \label{supp-eq:Tr} \\
    M_{\rm J} C_{p,w} \frac{dT_{\rm J}}{dt} &= F_{\rm cw}\,C_{p,w}(T_{\rm cold} - T_{\rm J}) + F_{\rm hw}\,C_{p,w}(T_{\rm hot} - T_{\rm J}) - UA(T_{\rm J} - T_{\rm r}) \label{supp-eq:TJ} \\
    \frac{dN_{\rm gA}}{dt} &= \frac{1000}{M_{\rm gA}} (F_{\rm A} - F_{\rm g2l}) \\
    \frac{dP}{dt} &= \frac{R}{V_{\rm g}} \left( T_{\rm r} \frac{dN_{\rm gA}}{dt} + N_{\rm gA} \frac{dT_{\rm r}}{dt} \right)
\end{align}
Here $R_{\rm poly}$ is the monomer consumption rate per unit liquid volume, $T_{\rm in}$ is the feed temperature, and $Q_{\rm loss}$ and $Q_{\rm stir}$ are heat loss and stirring power. $N_{\rm gA}$ is the gas-phase monomer inventory. The heat capacity in the implemented full model is
\[
 C_{\rm eff}=M C_{p,\rm mix}+(N_{\mu1}+N_{\lambda1})\frac{M_{gA}}{1000}C_{p,\rm poly}
             +N_{gA}^0\frac{M_{gA}}{1000}C_{p,A},
\]
with $M=4000$ kg, $C_{p,\rm poly}=1300$ J/(kg K), and initial gas-phase inventory $N_{gA}^0$. The reduced predictor instead uses
\[
 M C_{p,\rm mix}+\left[N_{gA}^0\frac{M_{gA}}{1000}+M_{A,\rm in}\right]C_{p,A}.
\] Here $M_{A,\rm in}$ is the cumulative monomer feed mass. The difference between these heat capacity expressions is retained as structural reduction error. The pressure equation follows the benchmark approximation, with $V_g=(6000-V_l)/1000$ in m$^3$ and without the gas volume displacement derivative. In the volume balance, $\dot V_l=0.06667F_{g2l}$ with $V_l$ in L and $F_{g2l}$ in kg/s.

The model feed names satisfy $F_{\rm A}=u_A$ and $F_{\rm B}=u_B$. Here $M_{\rm J}$ is the jacket water holdup (kg), $C_{p,w}$ the water specific heat capacity (J/(kg\,K)), $F_{\rm cw}$ and $F_{\rm hw}$ are the cold and hot jacket water flow rates (kg/s) determined by the applied utility input $v_{\rm sr}^{\rm act}\in[-1,1]$ after actuator dynamics:
\begin{align}
    F_{\rm cw} &= F_{\rm cw}^{\max}\;\max\!\bigl(0,\,-v_{\rm sr}^{\rm act}\bigr), &
    F_{\rm hw} &= F_{\rm hw}^{\max}\;\max\!\bigl(0,\,v_{\rm sr}^{\rm act}\bigr). \label{supp-eq:split-range}
\end{align}
$T_{\rm cold}$ and $T_{\rm hot}$ denote the cold and hot utility inlet temperatures, respectively.
The applied utility input $v_{\rm sr}^{\rm act}$ is distinct from the split-range command $v_{\rm sr}$ before actuator dynamics and the economic command $\alpha_{\rm eco}$.

The gas-liquid mass transfer is driven by the difference between partial pressure and Henry's law saturation:
\begin{equation}
    F_{\rm g2l} = k_{\rm mt} (P H - 1000 c_{\rm A})
\end{equation}

\subsection{Model Parameters}
The positive heat of reaction is $\Delta H=-\Delta H_{\rm rxn}$, consistent with the main manuscript. The nominal thermal parameters are $-\Delta H_{\rm rxn}=327240$ J/mol and $UA=21496$ W/K. The scenario parameters are:
\[
\begin{array}{c|rr}
 & -\Delta H_{\rm rxn}\ (\mathrm{J/mol}) & UA\ (\mathrm{W/K})\\\hline
N&327240&21496\\
\mathrm{PM}_{-}&261792&25795.2\\
\mathrm{PM}_{+},F&392688&17196.8
\end{array}
\]
Scenario scaling is applied to the plant parameters. Both NMPC variants start from the nominal model, with A-NMPC updating its parameter estimates from measurements. The plant and controller parameter vectors are recorded for every run.

\begin{table}[htbp]
	\centering
	\caption{Kinetic and physicochemical parameters used in validation simulations.}
	\label{supp-tab:para}
	\small
	\begin{tabularx}{\textwidth}{@{}lXll@{}}
		\toprule
		{Parameter} & {Description} & {Value} & {Units} \\ 
		\midrule
		$V_{tot}$         & Total reactor volume & 6000 & L \\
		$-\Delta H_{\rm rxn}$ & Heat of reaction & 327240 & J/mol \\
		$UA$              & Reactor-jacket thermal conductance & 21496 & W/K \\
		$C_{p,mix}$       & Specific heat capacity (Reacting Mix) & 4200 & J/(kg K) \\
		$C_{p,A}$         & Specific heat capacity (Feed A) & 804 & J/(kg K) \\
		$H$               & Henry's constant & $1/700$ & mol/(Pa m$^3$) \\
		$k_{mt}$          & Mass transfer coefficient & 0.0562 & kg m$^3$/(mol s) \\
		\midrule
		$k_{10}$          & Pre-exponential factor (decomposition) & 1.13E+17 & 1/s \\
		$k_{20}$          & Pre-exponential factor (initiation) & 3.62E+15 & L/(mol s) \\
		$k_{30}$          & Pre-exponential factor (propagation) & 5.49E+07 & L/(mol s) \\
		$k_{40}$          & Pre-exponential factor (transfer to monomer) & 9.96E+07 & L/(mol s) \\
		$k_{50}$          & Pre-exponential factor (transfer to agent) & 3.15E+06 & L/(mol s) \\
		$k_{70}$          & Pre-exponential factor (termination) & 3.38E+09 & L/(mol s) \\
		$E_1$             & Activation energy (decomposition) & 1.35E+05 & J/mol \\
		$E_2$             & Activation energy (initiation) & 119715 & J/mol \\
		$E_3$             & Activation energy (propagation) & 17413.76 & J/mol \\
		$E_4$             & Activation energy (transfer to monomer) & 53020 & J/mol \\
		$E_5$             & Activation energy (transfer to agent) & 20000 & J/mol \\
		$E_7$             & Activation energy (termination) & 13604.5 & J/mol \\
		\midrule
		$M_{gA}$          & Molar mass A & 100 & g/mol \\
		$M_{gB}$          & Molar mass B & 228 & g/mol \\
		\midrule
		$M_{\rm J}$       & Jacket water holdup & 2000 & kg \\
		$C_{p,w}$         & Water specific heat capacity & 4200 & J/(kg K) \\
		$T_{\rm cold}$    & Cold utility inlet & 283 & K \\
		$T_{\rm hot}$     & Hot utility inlet & 363 & K \\
		$F_{\rm cw}^{\max}$ & Max cold water flow & 10 & kg/s \\
		$F_{\rm hw}^{\max}$ & Max hot water flow & 1 & kg/s \\
		\bottomrule
	\end{tabularx}
\end{table}

\section{Mismatch between Plant and Model: The Gel Effect}
\label{supp-sec:mismatch}
To evaluate the ARC safety response, we introduce a plant--model mismatch with strong auto-acceleration (gel effect) in the plant. The free-volume model adds one auxiliary state and modifies the termination rate constant $k_7$ as follows.

\subsection{Free Volume and Conversion}
Define the monomer conversion $X$ and polymer volume fraction $\phi_p$:
\begin{align}
    X &= \frac{c_{A,0}\,V_{l,0} + M_{A,l}\,(1000/M_{gA}) - c_A\,V_l}{c_{A,0}\,V_l + M_{A,l}\,(1000/M_{gA})}, \\
    \phi_p &= \frac{(1+\varepsilon_v)\,X}{1+\varepsilon_v\,X}, \qquad
    \varepsilon_v = \frac{0.966471 - 1.164\times 10^{-3}\,T_{\rm r}}{1.19504 - 3.3\times 10^{-4}\,T_{\rm r}},
\end{align}
where $M_{A,l}$ is the cumulative mass of monomer A dissolved into the liquid phase.
The free volume fraction is
\begin{equation}
    V_f = 0.025 + a_p\,(T_{\rm r} - T_{\rm gp})\,\phi_p + a_m\,(T_{\rm r} - T_{\rm gm})\,(1 - \phi_p).
\end{equation}

\subsection{Gel Onset Criterion}
The Trommsdorff criterion compares a diffusion-weighted index $K$ against a threshold $K_3$:
\begin{align}
    K &= \left(\frac{M_{gA}\,(\mu_2+\lambda_2)}{\mu_1+\lambda_1}\right)^{0.5}\!\exp\!\left(\frac{p_A}{V_f}\right), &
    K_3 &= \exp\!\left(-0.4 + \frac{4460}{T_{\rm r}}\right).
\end{align}
When $K \ge K_3$, the system enters the gel regime.

\subsection{Modified Termination Rate}
The gel effect introduces an auxiliary tracking state $\overline{M}_{w,\mathrm{cr}}$ that latches the weight-average molecular weight at the onset of the gel transition:
\begin{equation}
    \frac{d\overline{M}_{w,\mathrm{cr}}}{dt} =
    \begin{cases}
        \displaystyle\frac{M_w(t) - \overline{M}_{w,\mathrm{cr}}}{\tau_{\mathrm{cr}}}, & K < K_3, \\[6pt]
        0, & K \ge K_3,
    \end{cases}
    \label{supp-eq:Mwcr_track}
\end{equation}
where $M_w(t) = M_{gA}\,(\mu_2+\lambda_2)/(\mu_1+\lambda_1)$ is the instantaneous weight-average molecular weight and $\tau_{\mathrm{cr}} = 10^{-4}$~s is a fast tracking time constant. The effective termination rate constant is then
\begin{equation}
    k_7^{\mathrm{eff}} =
    \begin{cases}
        \displaystyle k_7\,\exp\!\left(-p_A\,\max\!\left(\frac{1}{V_f}-\frac{1}{V_{f,\mathrm{cr}}},\,0\right)\right)
        \left(\frac{\overline{M}_{w,\mathrm{cr}}}{M_w(t)}\right)^{1.75}, & K \ge K_3, \\[8pt]
        k_7, & K < K_3.
    \end{cases}
    \label{supp-eq:gel_kt}
\end{equation}
The first factor captures the free-volume diffusion limitation, and the second factor introduces a molecular-weight-dependent auto-acceleration: as $M_w$ grows beyond the onset value $\overline{M}_{w,\mathrm{cr}}$, termination is further suppressed, creating a positive feedback loop that can trigger thermal runaway.

\subsection{Gel Effect Parameters}
\begin{center}
\begin{tabular}{llll}
\toprule
\textbf{Symbol} & \textbf{Description} & \textbf{Value} & \textbf{Units} \\
\midrule
$V_{f,\mathrm{cr}}$ & Critical free volume & 0.129 & -- \\
$p_A$ & Diffusion exponent & 1.11 & -- \\
$T_{\rm gp}$ & Glass transition (polymer) & 383 & K \\
$T_{\rm gm}$ & Glass transition (monomer) & 167 & K \\
$a_p$ & Expansion coefficient (polymer) & $4.8\times 10^{-4}$ & K$^{-1}$ \\
$a_m$ & Expansion coefficient (monomer) & $1.0\times 10^{-3}$ & K$^{-1}$ \\
$\tau_{\mathrm{cr}}$ & Tracking time constant & $10^{-4}$ & s \\
\bottomrule
\end{tabular}
\end{center}

\section{Control Implementation Details}
\label{supp-sec:control_impl}

\subsection{ARC Anti-Windup and Signal Conventions}
The one-second implementation first computes monomer feed from the previously available pressure setpoint, then updates the temperature cascade, economic command and pressure setpoint for the next step. The cascade is enabled from the second sample. Its startup temperature setpoint is $T_0^s+(351-T_0^s)(10\theta^3-15\theta^4+6\theta^5)$ with $\theta=\min(t/3300,1)$ and $T_0^s$ equal to the first measured temperature plus 5 K. The initial true reactor temperature is 340 K. This benchmark-specific startup temperature setpoint is followed by the constant 351 K hold; NMPC retains its constant 351 K objective from startup.

The pressure integral state is held until pressure first reaches its active setpoint, then projected to the monomer input range. Monomer dose completion forces monomer feed to zero. The economic integral state is held after a feed-dose limit is reached, while the remaining permitted channel uses the current proportional demand. The economic command is limited to $[-40,10]$ in every mode and pressure setpoint relaxation is bounded below by 1.1 MPa. These overrides are distinct from the continuous projected-PI regional analysis.

The jacket temperature slave uses raw output $a=2(T_{J,\rm sp}-T_J)$ and the mapping
\[
 \ell_J(a)=a[1-3(1-\tanh a)/7],\qquad
 \mathcal S_J(a)=\ell_J(a)[1-(1-\tanh(\ell_J(a)+1))/4].
\]
The split-range command is clipped to $[-1,1]$ before actuator dynamics. Virtual split-range demand evaluates the same mapping from the reactor temperature master setpoint before feasibility limits and uses the separate numerical guard $[-11,1]$. The reactor temperature master setpoint limits are obtained by inverting this nonlinear mapping at the physical utility limits and intersecting the result with $[280,400]$ K.

The ARC implementation uses two anti-windup mechanisms. The temperature cascade uses back-calculation anti-windup:
\begin{align}
    u_{I,k} &= u_{I,k-1} + \frac{K_p T_s}{T_i} e_k + \frac{T_s}{T_t} (u_{sat, k} - u_{calc, k}) \\
    u_{sat, k} &= \text{clip}(u_{calc, k}, u_{min}, u_{max})
\end{align}
This structure is used in the tracking loop because the jacket temperature setpoint and split-range command can saturate during startup and large disturbances.

The economic PI controller bounds its integral state by explicit clamping:
\begin{align}
    r_k &= \text{clip}\!\left(r_{k-1} + \frac{K_{\rm VPC} T_s}{T_{i,{\rm VPC}}} e_{{\rm VPC},k},\ -r_-,\ r_+\right), \\
    \alpha_{{\rm eco},k} &= \text{clip}\!\left(K_{\rm VPC} e_{{\rm VPC},k} + r_k,\ -40,\ 10\right).
\end{align}
This discrete integral projection uses the same integral state $r$ as the main text. The normalized cooling demand error is
\[
e_{\rm VPC} = v_{\rm sr}^{\rm v} - v_{{\rm sr},sp},
\]
where $v_{\rm sr}^{\rm v}$ is virtual split-range demand before master and utility limits. On the cooling branch, $e_{\rm VPC}=e_v/F_{\rm cw}^{\max}$, so both errors have the same physical sign. The economic command $\alpha_{\rm eco}$ drives initiator feed and pressure setpoint relaxation through the static maps below.

Let $v_k:=\alpha_{{\rm eco},k}$, $[v]_+ := \max(v,0)$, and $[v]_- := \min(v,0)$. The ARC implementation uses the static maps
\begin{align}
    u_{B,k}^{\rm fb} &= \text{clip}\!\left(10^{-7}[v_k]_+,\ 0,\ u_B^{\max}\right), \\
    P_{{\rm sp},k} &=
    \begin{cases}
    \text{clip}\!\left(P_{\rm sp}^{\rm nom}+10^4[v_k]_-,\ 1.1~\mathrm{MPa},\ P_{\rm sp}^{\rm nom}\right), & \text{cascade enabled},\\
    P_{\rm sp}^{\rm nom}, & \text{otherwise},
    \end{cases}
\end{align}
with $P_{\rm sp}^{\rm nom}=1.5~\mathrm{MPa}$ and $v_k$ already clipped to $[-40,10]$ by the VPC update. Thus positive economic command increases the feedback initiator feed, whereas negative economic command lowers the pressure setpoint. After the initiator dose limit is reached, $u_B$ is set to zero. The implemented actuator commands are then passed through the actuator lags and rate limits used in the case-study simulation.

\section{Reduced-benchmark selector and feasibility proof}
\label{supp-sec:supp-reduced-selector}
For the isothermal A--B benchmark, $N_A,N_B$ are reactant inventories, $V$ is liquid volume, $u$ is B-feed rate, $c_B=N_B/V$ and $c_{B,f}$ is feed concentration. The mass balances are $\dot N_A=-R_{\rm mol}$, $\dot N_B=c_{B,f}u-R_{\rm mol}$ and $\dot V=u$, with $R_{\rm mol}=kN_AN_B/V$. The cooling failure temperature error is $e_v=T_{\max}-T-C\min(N_A,N_B)/V-e_{\rm sp}$, where $C=\Delta H/(\rho c_p)$. The PI request is $u_{\rm PI}=\operatorname{sat}(K_Pe_v+r;0,u_{\max})$, with $r\in[-r_-,r_+]$. On the B-limited branch, $\dot e_v=d_0-qu$, where $d_0=Ck(N_A/V)c_B$ and $q=C(c_{B,f}-c_B)/V$. At the recovery edge $e_v=\eta_e$, $A_e^- =\operatorname{sat}(K_P\eta_e-r_-;0,u_{\max})$ is the minimum PI request over all integral states.

The benchmark tests two complementary thermal constraints: the temperature reached after cooling failure and the instantaneous cooling capacity during normal operation. The projected PI regulates the cooling failure temperature error. We augment it with an algebraic cooling capacity selector to enforce the second constraint and evaluate both controllers against the same constraint set. Write $R_{\rm mol}=kN_AN_B/V$ for molar reaction rate, $Q_{\max}(V)=U(2V/\sigma+\pi\sigma^2)(T-T_{c,\min})$, and $h_Q=Q_{\max}-\Delta HR_{\rm mol}$. Differentiation gives
\begin{align}\label{supp-eq:heat-selector-coefficients}
 \frac{d}{dt}(\Delta HR_{\rm mol}-Q_{\max})&=a_Q+b_Qu,\\
 a_Q&=-\Delta H kR_{\rm mol}(N_A+N_B)/V,\notag\\
 b_Q&=\Delta H k(N_A/V)(c_{B,f}-c_B)-Q'_{\max}(V).\notag
\end{align}
For $b_Q>0$, define
\begin{equation}\label{supp-eq:heat-selector}
 u_Q=\operatorname{sat}\!\left(\frac{-a_Q+\kappa_Q h_Q}{b_Q};0,u_{\max}\right),
 \qquad u=\min\{u_{\rm PI},u_Q\},
\end{equation}
with $\kappa_Q=0.2$ s$^{-1}$ and projected-PI request $u_{\rm PI}$. If $b_Q\le0$, set $u_Q=u_{\max}$. At full volume, the metered volume shutoff sets $u=0$. This selector is evaluated from the exact reduced model heat balance.

\begin{proposition}[Feasibility of the reduced implementation]\label{supp-prop:supp-reduced-feasibility}
For the exact reduced model, suppose $h_Q(0)\ge0$, $V(0)\le V_{\max}$, $e_v(0)\ge-\eta_s$, $e_{\rm sp}\ge\eta_s$ and $K_P\eta_s\ge r_+$. Then the projected PI with~\eqref{supp-eq:heat-selector} and volume shutoff preserves the feed, volume, instantaneous heat removal and cooling failure temperature constraints in~the reduced benchmark.
\end{proposition}
\noindent\emph{Proof.} The mass balances preserve nonnegative reactant inventories. Since $a_Q\le0$, the selector gives $\dot h_Q\ge-\kappa_Qh_Q$ when $b_Q>0$ and $h_Q\ge0$; for $b_Q\le0$, $\dot h_Q\ge0$. Thus $h_Q$ cannot cross below zero. At $e_v=-\eta_s$ on the B-limited branch, the integral bound gives $u_{\rm PI}=0$ and hence $\dot e_v=d_0\ge0$. On the A-limited branch, $T_{\rm cf}=T+CN_A/V$ is nonincreasing because A is consumed and $V$ is nondecreasing. Therefore $e_v\ge-\eta_s$ and $T_{\rm cf}\le T_{\max}-e_{\rm sp}+\eta_s\le T_{\max}$. Saturation and the volume shutoff give the remaining bounds. \hfill$\square$

The selector also retains an economic recovery margin in the regional band. Let $c_{B,e}$ denote the concentration of B at the upper error boundary $e_v=\eta_e$. The feed that holds the error constant there is $u_{\rm eq}=d_0/q=kN_Ac_{B,e}/(c_{B,f}-c_{B,e})$. For $b_Q>0$, $h_Q\ge0$ and $Q'_{\max}\ge0$,
\begin{equation}\label{supp-eq:heat-selector-recovery}
 q(u_Q^{\rm raw}-u_{\rm eq})
 =\frac{q[\Delta H kR_{\rm mol} c_{B,e}+Q'_{\max}u_{\rm eq}+\kappa_Qh_Q]}{b_Q}
 \ge Ckc_{B,e}^{2},
\end{equation}
where $u_Q^{\rm raw}=(-a_Q+\kappa_Qh_Q)/b_Q$. Saturating this request at $u_{\max}$ cannot reduce the actual upper-edge feed below $\min\{A_e^-,u_Q^{\rm raw}\}$. The regional inward margin is consequently the smaller of the projected-PI margin and $Ckc_{B,e}^2=1.1094\times10^{-3}$ K/s. The same conclusion holds when $b_Q\le0$ because the selector is inactive. Thus the additional physical constraint preserves the first-order band argument on the declared region, while Proposition~\ref{supp-prop:supp-reduced-feasibility} establishes nominal feasibility beyond that regional economic-band statement.

\subsection{Numerical reference and integration checks}
The numerical reference uses CasADi/IPOPT direct transcription with fourth-order Runge--Kutta dynamics. Refinement from 100 to 400 intervals gives terminal conversions 0.9219 and 0.9223. The 100-interval feed has a small positive cooling capacity residual, $5.16\times10^{-3}$ W, under independent integration; only the refined reference is used for the feasible-performance comparison. Small numerical guards of $10^{-3}$ K and $10^{-3}$ W are included in the transcription; independent adaptive integration of the 400-interval feed gives a peak cooling failure temperature of 396.9990 K and a maximum cooling capacity residual of $-7.0740\times10^{-4}$ W.

Projected PI with the cooling capacity selector achieves terminal conversion 0.9199, a gap of 0.2411 percentage points from the refined reference. Over the recorded regional interval, 2.0408--25.3611 h, $e_v$ lies between $-0.04552$ and $0.04939$ K, inside the selected band. The largest recorded cooling capacity residual is $3.25\times10^{-10}$ W, within the $10^{-6}$ W numerical feasibility tolerance. The full-trajectory peak $T_{\mathrm{cf}}$ is 396.0455 K. Refining the closed-loop integration changes terminal conversion by $3.67\times10^{-8}$. The computation verifies the exact first-order derivative by an independent directional-difference check and records the region-exit event.

\section{Supporting local response calculations}\label{supp-sec:supp-response-checks}

\subsection{Local jacket model and retained gain}\label{supp-sec:supp-jacket-check}
The jacket calculation uses the independent nominal ARC design trajectory, seed 90. At 7200 s, the saved state is $T_r=351.004$ K, $T_J=307.831$ K. Automatic differentiation of the full nonlinear plant right-hand side with respect to the state and actual applied utility gives, with the remaining states frozen,
\begin{equation}\label{supp-eq:jacket-local-model}
 \delta\dot T_J=-\frac{\delta T_J}{142.897}+0.124157\,\delta v_{\rm sr}^{\rm act},\qquad
 G_J(s)=\frac{\delta T_J(s)}{\delta v_{\rm sr}^{\rm act}(s)}
       =\frac{17.7417}{142.897s+1}.
\end{equation}
Here $\delta$ denotes a local perturbation and $s$ is the Laplace variable. The input is applied utility input, so the static gain is 17.7417 K per unit applied utility input. At 3600 s the same calculation gives input slope 0.140467 K/s and time constant 155.882 s. Jacket-to-reactor sensitivities are respectively $1.1561\times10^{-3}$ and $1.2454\times10^{-3}$ s$^{-1}$ at 7200 and 3600 s. 

The slave scaling relates this process model to the implemented gain. Write its raw output as $a_J=K_c(T_{J,\rm sp}-T_J)$ and its utility request as $\mathcal S_J(a_J)$, where
\begin{align}\label{supp-eq:slave-scaling}
 \ell_J(a)&=a\left[1-\frac37(1-\tanh a)\right],\notag\\
 \mathcal S_J(a)&=\ell_J(a)\left[1-\frac14(1-\tanh(\ell_J(a)+1))\right].
\end{align}
At 7200 s, $\mathcal S'_J=0.05347$. The retained raw gain $K_c=2$ gives effective proportional gain $K_{J,\rm eff}=0.10694$ K$^{-1}$ and ideal local P-loop response time $142.897/(1+17.7417\times0.10694)=49.3$ s. The corresponding value at 3600 s is 39.0 s. Conversely, choosing a 45 s response in this approximation gives raw gains 2.29 and 1.65 at those two states, respectively. The retained gain therefore gives a local response close to this illustrative target.

\subsection{Nominal and parameter-adaptive output-feedback NMPC}
\label{supp-sec:supp-new-nmpc}
Both variants use the same twelve-state reduced NMPC model, objective, measurements, horizon, input bounds and recovery procedure. The state order is
\[
 x=[N_{\lambda0},N_R,N_A,N_B,N_C,T_r,N_{gA},P,T_J,M_{A,\mathrm{in}},M_{B,\mathrm{in}},V_l]^\top.
\]
Here $N_R$ is the free-radical inventory, $N_A,N_B,N_C$ are liquid species inventories, and $M_{A,\mathrm{in}},M_{B,\mathrm{in}}$ are cumulative feed masses. The full plant retains higher polymer moments and their heat capacity contribution; the gel effect plant has a further state. Thus even N contains structural reduction error.

\subsubsection{Physical parameter augmentation}
In the adaptive model, $\theta_H$ and $\theta_{UA}$ are the ratios of the heat of reaction and reactor--jacket thermal conductance to their respective nominal values:
\[
 \Delta H=327240\theta_H\ \mathrm{J/mol},\quad UA=21496\theta_{UA}\ \mathrm{W/K},\quad
 \dot\theta_H=\dot\theta_{UA}=0.
\]
The estimates of $\Delta H$ and $UA$ enter all corresponding physical balance terms: the same estimated $UA$ appears in reactor and jacket balances, and pressure retains its dependence on the reactor temperature derivative. A fourteen-state extended Kalman filter (EKF) jointly estimates the twelve process states and the ratios $\theta_H$ and $\theta_{UA}$. The nominal comparator uses the same process state EKF with $\Delta H$ and $UA$ fixed at their nominal values. Neither EKF receives the actual mismatch parameters. In the fault scenario, both the predictor and EKF retain the nominal reaction kinetics, without a gel model or gel state; A-NMPC continues to estimate $\Delta H$ and $UA$ from measurements. The parameter estimates are fixed within each prediction horizon and updated from subsequent measurements.

\subsubsection{Prediction problem and constraint recovery}
The controller solves every 30 s with a 30 min horizon, 60 intervals and degree-three Radau collocation. Its stage objective retains the original production and pressure terms while assigning reactor temperature regulation priority during the prescribed hold:
\[
 L=-R_A+\frac{3\times10^5}{360}\left(\frac{P}{1.5\times10^6}-1\right)^2\omega(M_{A,\mathrm{in}})
       +q_T(t)(T_r-351)^2,
\]
\[
 \omega(M)=\{1+\exp[-10(3250-M)+10]\}^{-1}.
\]
Here $R_A$ is the total monomer consumption rate in kg/h, and $q_T(t)=3\times10^2/360$ before 3300 s and $q_T(t)=q_T^{\rm hold}=833.333$ thereafter, using absolute predicted time at each collocation point. The hold weight is selected on separate design data and shared by both NMPC variants.
Only pressure tracking is switched off near charge completion. Input-move regularization adds $\sum_k[8,0.1,4](u_k-u_{k-1})^{\circ2}$ in physical input units. The integrated stage objective and move regularization are multiplied by the same numerical factor $10^{-3}$. This common scaling leaves their relative weights unchanged.

The first solve imposes $T_r\le351.7$ K and $P\le1.6$ MPa at all future collocation and interval-end states, and additionally $T_r\ge350.3$ K when their absolute time is at least 3300 s. Physical input limits and dose caps 3250 kg and 0.002 kg also apply. The startup lower numerical temperature bound is 320.15 K. The current estimate is imposed by equality and is not clipped to the future operating temperature limit. If it already violates an active operating or recipe bound, or the first solve fails acceptance, a recovery problem permits common horizon slacks:
\[
 350.3-\epsilon_T\le T_r\le351.7+\epsilon_T,\qquad P/10^6\le1.6+\epsilon_P,
\]
\[
 0\le\epsilon_T\le21.45\ \mathrm{K},\qquad 0\le\epsilon_P\le0.4\ \mathrm{MPa}.
\]
The displayed lower recovery bound applies on the hold; before it, the lower numerical state bound remains 320.15 K. The recovery objective adds $10^5(\epsilon_T+\epsilon_T^2)+10^7(\epsilon_P+\epsilon_P^2)$ using numerical slacks expressed in K and MPa. This feasibility penalty is not multiplied by $10^{-3}$. Input and dose limits remain enforced, and predicted temperature remains bounded by 373.15 K. The 2.0 MPa bound limits the numerical search domain; the operating pressure limit remains 1.6 MPa. The recovery problem computes a return from an estimated violation. Both NMPC variants use this procedure.

CasADi 3.8.0/IPOPT uses exact derivatives, an expanded graph, at most 500 iterations per attempt, tolerance $10^{-6}$ and acceptable tolerance $10^{-5}$. An iterate is applied only with successful or acceptable convergence, finite decision variables, and maximum normalized constraint/bound violation below $10^{-6}$. A failed first attempt is retried with a reinitialized recovery problem. If neither attempt is accepted, zero monomer and initiator feed with maximum cooling is commanded, retaining the physical actuator dynamics. All attempts, residuals, slacks, exit flags, elapsed times and fallback actions are logged. Accepted primal trajectories are shifted between solves; future feed is fixed to zero once its metered dose is complete.

The reduced predictor omits short actuator lags and rate limits; the EKF propagates with actual applied input, and all full plants retain those dynamics. The differentiable utility crossover is $[v]_{+,\epsilon}=(v+\sqrt{v^2+\epsilon^2})/2$, $\epsilon=10^{-3}$; the plant retains its exact split range. Separate design trajectories check this approximation before evaluation.

\subsubsection{EKF and normalization}
Both EKFs update their estimates every second from measured $[T_r,T_J,P]$. The measurement noise is independent and Gaussian, with standard deviations $[0.1/3,0.1/3,5000/3]$ in K, K and Pa. Physical measurement covariance is therefore
\[
 R_y=\operatorname{diag}((0.1/3)^2,(0.1/3)^2,(5000/3)^2).
\]
Each current measurement update uses the Joseph covariance formula and is followed by one-second state/covariance propagation with the actual applied input. Exact metered doses overwrite their estimates, with covariance rows and columns zeroed. The twelve physical process noise standard deviations per propagation are
\[
 [10^{-9},10^{-5},0.02,10^{-8},10^{-4},0.002,0.001,10,0.002,0,0,0.001].
\]
Initial physical state standard deviations are
\[
 [10^{-7},10^{-4},1,10^{-6},0.05,0.03,0.01,1000,0.03,0,0,0.1].
\]
The ratios $\theta_H$ and $\theta_{UA}$ are initialized at one, so estimation of $\Delta H$ and $UA$ starts from their nominal values. In A-NMPC, each ratio has initial standard deviation 0.2 and random-walk standard deviation $10^{-4}$ per one-second propagation. In nominal NMPC, the ratios remain at one and their initial and random-walk standard deviations are zero. The ratio estimates are projected onto $[0.5,1.5]$, with projection indicators saved at the 30 s controller updates. The same bounds apply in every scenario. Positivity and physical lower bounds are applied to estimated process states.

Normalization is
\[
 \begin{aligned}s_x=[&1.5\times10^{-4},0.05,9000,4\times10^{-4},10,365,\\
 &1075,1554000,350,3250,0.002,4225].\end{aligned}
\]
The dimensionless ratios $\theta_H$ and $\theta_{UA}$ use unit scales; input scales are $[0.694,4\times10^{-7},1]$. Physical standard deviations are divided by their scales and squared to form normalized covariances. Source, serialized models, configurations, saved noise streams and raw trajectories accompany the reproduction package.

\subsubsection{Independent design checks}
Before evaluation, both EKFs are replayed on separate ARC design trajectories (seed 90) in N and PM$_+$. In this offline check, each 30 s prediction starts from the posterior estimate and uses fixed parameter estimates and recorded applied inputs, with no later measurements. Over 3300--5999 s in PM$_+$, adaptation reduces reactor temperature prediction root mean square error (RMSE) from 0.2646 to 0.01734 K and jacket temperature RMSE from 0.9541 to 0.00479 K. At 6000 s, the estimated ratios $(\theta_H,\theta_{UA})$ are approximately $(1.1993,0.8001)$, compared with the imposed ratios $(1.2,0.8)$. In N, the estimated ratios are $(0.9995,1.0000)$ and reactor temperature RMSE stays approximately 0.016 K for either EKF. The estimates improve prediction on these design data. Under structural model error or the gel effect, adaptation can yield effective parameters for the reduced model.

Equation checks compare the parameterized and serialized reduced balances; finite differences check the flow Jacobian, and repeated EKF steps check covariance symmetry and positive semidefiniteness. Separate complete design batches check convergence, recovery and the actuator approximation. Fresh paired evaluation seeds 200--204 are run after freezing the implementation and settings.

Two hold weights, 833.333 and 8333.333, were compared before evaluation. In the separate N and F adaptive design cases, both meet the full temperature recipe and operating upper bounds. The smaller weight gives temperature MAEs of 0.0245 and 0.0444 K, compared with 0.0599 and 0.0607 K, and batch times of 3.702 and 4.102 h, compared with 3.815 and 4.121 h. It is therefore retained for both NMPC variants.

Table~\ref{supp-tab:supp-quality-design} retains all ten seed-191 design outcomes. The eight higher-weight cases check both nominal and adaptive variants across the four scenarios; the separate lower-weight check tests the adaptive variant in N and F. These design observations are excluded from the fresh five-seed evaluation. Physical actuator/dose recurrences and accepted-iterate residuals were independently verified before freezing the chosen settings. In the table, C/S indicates physical completion and completion meeting the specifications. Time is the physical stopping time; temperature errors use the hold interval. Fallback counts updates for which both solver attempts failed.
\begin{table}[!t]\centering\footnotesize
\caption{Temperature-priority design runs, seed 191. MAE: mean absolute error; --: no hold data.}\label{supp-tab:supp-quality-design}
\setlength{\tabcolsep}{3pt}
\begin{tabular}{llccccc}\toprule
Scen. & Control & $q_T^{\rm hold}$ & C/S & Time (h) & $T$ MAE / max (K) & Fallback\\\midrule
N & NMPC & 8333.3 & 1/1 & 3.946 & 0.083 / 0.518 & 7\\
N & A-NMPC & 8333.3 & 1/1 & 3.815 & 0.060 / 0.350 & 0\\
PM$_-$ & NMPC & 8333.3 & 1/1 & 4.184 & 0.256 / 0.602 & 3\\
PM$_-$ & A-NMPC & 8333.3 & 1/1 & 3.698 & 0.097 / 0.525 & 0\\
PM$_+$ & NMPC & 8333.3 & 0/0 & 0.926 & 19.969 / 22.206 & 17\\
PM$_+$ & A-NMPC & 8333.3 & 1/1 & 4.431 & 0.063 / 0.397 & 1\\
F & NMPC & 8333.3 & 0/0 & 0.776 & -- / -- & 12\\
F & A-NMPC & 8333.3 & 1/1 & 4.121 & 0.061 / 0.262 & 3\\
N & A-NMPC & 833.3 & 1/1 & 3.702 & 0.025 / 0.151 & 0\\
F & A-NMPC & 833.3 & 1/1 & 4.102 & 0.044 / 0.287 & 1\\
\bottomrule\end{tabular}

\end{table}

\begin{figure}[!t]\centering
\includegraphics[width=\textwidth,height=0.75\textheight,keepaspectratio]{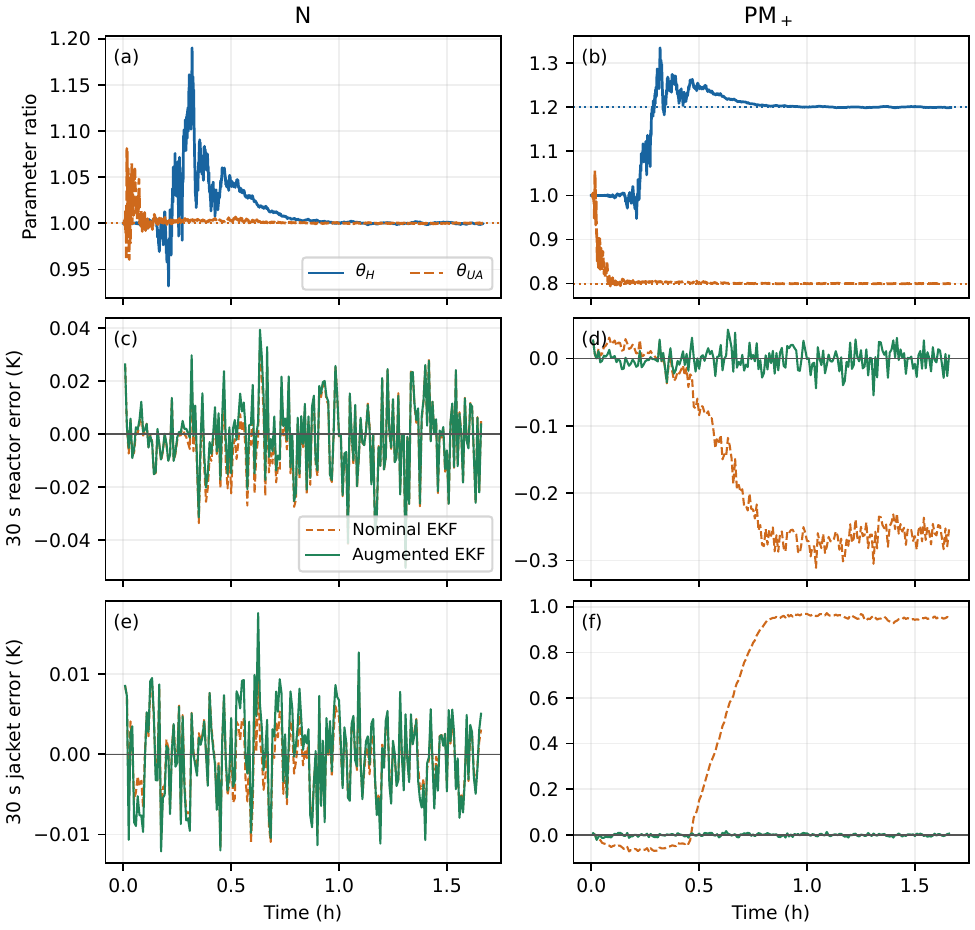}
\caption{Independent EKF design check, ARC seed 90. Top: estimated parameter ratios $\theta_H$ and $\theta_{UA}$, with dotted lines showing the imposed ratios. Middle and bottom: 30 s conditional prediction error for reactor and jacket temperatures. Predictions use recorded applied inputs and no future measurements. N is shown on the left and PM$_+$ on the right; axis scales differ to resolve nominal errors.}
\label{supp-fig:supp-parameter-validation}
\end{figure}

\subsubsection{Physical execution and reporting}
The full plant is integrated every second. Actuator time constants are $[5,5,3]$ s and per-second slew limits $[0.0347,2\times10^{-8},0.05]$. The Euler actuator update uses saturated commands, clipped increments and final saturation. Subsequent dose clipping represents an ideal metered shutoff that can override normal slew. All controllers start at applied input $[0.694,0,0]$.

Completion requires monomer charge 3250 kg within 1 g and true pressure $P\le1.0$ MPa without a temperature trip. The trip is $T_r\ge373.15$ K and maximum duration is 30000 s. Physical completion, temperature recipe conformity, operating limit exceedance and trip are distinct outcomes. Primary batch time comparisons require both completion and conformity with the hold recipe and operating upper bounds. Failed or incomplete runs remain in the counts. Endpoint properties use completed batches' final full model states; all time histories include the terminal physical sample.

\section{Additional Industrial Evaluation Results}
\label{supp-sec:supp-industrial-results}
The tables use the same 60 main runs and true terminal states as the main manuscript. Let $N_{\lambda j}$ and $N_{\mu j}$ be living- and dead-chain moment totals. The number-average molecular weight $M_n$, weight-average molecular weight $M_w$ and polydispersity index (PDI) are
\[
 M_n=100\frac{N_{\lambda1}+N_{\mu1}}{N_{\lambda0}+N_{\mu0}},\qquad
 M_w=100\frac{N_{\lambda2}+N_{\mu2}}{N_{\lambda1}+N_{\mu1}},\qquad
 \mathrm{PDI}=M_w/M_n,
\]
with $M_n,M_w$ in g/mol. Polymer mass is $0.1(N_{\lambda1}+N_{\mu1})$ kg. Property means include completed batches only; per-run values and ranges are retained in the machine-readable results. Grade-specific limits determine product acceptance.
\begin{table}[!t]\centering\footnotesize
\caption{Endpoint property means of completed batches. Each controller/scenario has five prescribed runs; $n$ is the completed count. Molecular weights are scaled by $10^5$ g/mol.}\label{supp-tab:supp-properties}
\setlength{\tabcolsep}{3pt}
\begin{tabular}{llccccc}\toprule
Scen. & Control & $n$ & Polymer (kg) & $M_n/10^5$ & $M_w/10^5$ & PDI\\\midrule
N & ARC & 5 & 2680.0 & 31.833 & 63.807 & 2.004\\
N & NMPC & 5 & 2680.1 & 31.750 & 63.640 & 2.004\\
N & A-NMPC & 5 & 2680.1 & 31.752 & 63.644 & 2.004\\
PM$_-$ & ARC & 5 & 2680.0 & 31.829 & 63.798 & 2.004\\
PM$_-$ & NMPC & 5 & 2680.2 & 31.464 & 63.080 & 2.005\\
PM$_-$ & A-NMPC & 5 & 2680.0 & 31.810 & 63.768 & 2.005\\
PM$_+$ & ARC & 5 & 2680.1 & 31.813 & 63.771 & 2.005\\
PM$_+$ & NMPC & 0 & -- & -- & -- & --\\
PM$_+$ & A-NMPC & 5 & 2680.1 & 31.709 & 63.556 & 2.004\\
F & ARC & 5 & 2680.8 & 30.235 & 60.528 & 2.002\\
F & NMPC & 0 & -- & -- & -- & --\\
F & A-NMPC & 5 & 2680.8 & 30.477 & 61.020 & 2.002\\
\bottomrule\end{tabular}

\end{table}\begin{table}[!t]\centering\footnotesize
\caption{Pressure constraint check over all true one-second samples, including incomplete runs. The operating upper bound is 1.6 MPa.}\label{supp-tab:supp-pressure}
\setlength{\tabcolsep}{3pt}
\begin{tabular}{llcc}\toprule
Scen. & Control & Maximum peak (MPa) & Exceeding runs\\\midrule
N & ARC & 1.5016 & 0/5\\
N & NMPC & 1.5424 & 0/5\\
N & A-NMPC & 1.5414 & 0/5\\
PM$_-$ & ARC & 1.5016 & 0/5\\
PM$_-$ & NMPC & 1.5690 & 0/5\\
PM$_-$ & A-NMPC & 1.5668 & 0/5\\
PM$_+$ & ARC & 1.5054 & 0/5\\
PM$_+$ & NMPC & 1.6185 & 4/5\\
PM$_+$ & A-NMPC & 1.5403 & 0/5\\
F & ARC & 1.5011 & 0/5\\
F & NMPC & 1.6161 & 5/5\\
F & A-NMPC & 1.5437 & 0/5\\
\bottomrule\end{tabular}

\end{table}\begin{table}[!t]\centering\footnotesize
\caption{Largest per-run percentage of hold samples outside $351.0\pm0.7$ K and below 350.3 K. Samples start at 3300 s; -- denotes five empty hold windows.}\label{supp-tab:supp-recipe-band}
\setlength{\tabcolsep}{3pt}
\begin{tabular}{llcc}\toprule
Scen. & Control & Outside recipe band (\%) & Below 350.3 K (\%)\\\midrule
N & ARC & 0.00 & 0.00\\
N & NMPC & 0.00 & 0.00\\
N & A-NMPC & 0.00 & 0.00\\
PM$_-$ & ARC & 0.00 & 0.00\\
PM$_-$ & NMPC & 2.51 & 2.51\\
PM$_-$ & A-NMPC & 0.00 & 0.00\\
PM$_+$ & ARC & 0.00 & 0.00\\
PM$_+$ & NMPC & -- & --\\
PM$_+$ & A-NMPC & 0.00 & 0.00\\
F & ARC & 0.00 & 0.00\\
F & NMPC & -- & --\\
F & A-NMPC & 0.00 & 0.00\\
\bottomrule\end{tabular}

\end{table}\begin{table}[!t]\centering\footnotesize
\caption{A-NMPC state estimates: worst per-run RMSE across five seeds, using estimates recorded every 30 s from 3300 s. Bound count totals updates at either parameter bound over all five full trajectories.}\label{supp-tab:supp-observer}
\setlength{\tabcolsep}{3pt}
\begin{tabular}{lcccc}\toprule
Scen. & $T_r$ RMSE (K) & $T_J$ RMSE (K) & $P$ RMSE (Pa) & Bound count\\\midrule
N & 0.0084 & 0.0064 & 165.9 & 0\\
PM$_-$ & 0.0080 & 0.0063 & 180.1 & 0\\
PM$_+$ & 0.0087 & 0.0064 & 153.8 & 0\\
F & 0.0086 & 0.0064 & 152.6 & 0\\
\bottomrule\end{tabular}

\end{table}

\subsection{Computation and accepted-iterate diagnostics}
ARC timings cover measurement processing, controller logic and the actuator update, excluding plant integration. NMPC timings cover the NLP solve alone, excluding its EKF and plant integration. The reported quantiles pool all recorded updates in each controller/scenario group. Cases ran concurrently on the same local Windows host under varying load. The software is MATLAB R2026a for ARC and Python 3.12.14/CasADi 3.8.0 for NMPC. All closed-loop simulations impose zero computation delay.
\begin{table}[!t]\centering\footnotesize
\caption{Elapsed computation time per update, pooled across five seeds.}\label{supp-tab:supp-timing}
\setlength{\tabcolsep}{3pt}
\begin{tabular}{llrrrr}\toprule
Scen. & Control & Median & 95th pct. & Maximum & Unit\\\midrule
N & ARC & 66.00 & 92.40 & 37380.50 & $\mu$s\\
N & NMPC & 3.32 & 14.70 & 149.95 & s\\
N & A-NMPC & 3.21 & 13.72 & 123.53 & s\\
PM$_-$ & ARC & 65.50 & 90.90 & 58588.70 & $\mu$s\\
PM$_-$ & NMPC & 3.41 & 15.51 & 157.94 & s\\
PM$_-$ & A-NMPC & 2.84 & 12.71 & 85.83 & s\\
PM$_+$ & ARC & 64.70 & 89.60 & 14292.70 & $\mu$s\\
PM$_+$ & NMPC & 4.64 & 89.73 & 276.80 & s\\
PM$_+$ & A-NMPC & 2.87 & 14.23 & 32.11 & s\\
F & ARC & 66.40 & 91.10 & 31267.50 & $\mu$s\\
F & NMPC & 3.41 & 91.03 & 185.60 & s\\
F & A-NMPC & 3.03 & 16.21 & 74.90 & s\\
\bottomrule\end{tabular}

\end{table}\begin{table}[!t]\centering\footnotesize
\caption{NMPC update counts: hard-constrained solution, recovery solution or fallback after two unsuccessful attempts. The last column counts updates exceeding 30 s, including retries.}\label{supp-tab:supp-solver}
\setlength{\tabcolsep}{3pt}
\begin{tabular}{llrrrrr}\toprule
Scen. & Control & Calls & Hard & Recovery & Fallback & $>30$ s\\\midrule
N & NMPC & 2230 & 2222 & 2 & 6 & 9\\
N & A-NMPC & 2226 & 2222 & 3 & 1 & 10\\
PM$_-$ & NMPC & 2316 & 2287 & 20 & 9 & 33\\
PM$_-$ & A-NMPC & 2212 & 2206 & 6 & 0 & 6\\
PM$_+$ & NMPC & 519 & 360 & 89 & 70 & 75\\
PM$_+$ & A-NMPC & 2542 & 2542 & 0 & 0 & 2\\
F & NMPC & 466 & 331 & 72 & 63 & 61\\
F & A-NMPC & 2454 & 2453 & 0 & 1 & 6\\
\bottomrule\end{tabular}

\end{table}
Figure~\ref{supp-fig:supp-utility} shows that the applied utility remains bounded while virtual split-range demand continues to indicate cooling overload.

\begin{figure}[!t]\centering
\includegraphics[width=\textwidth]{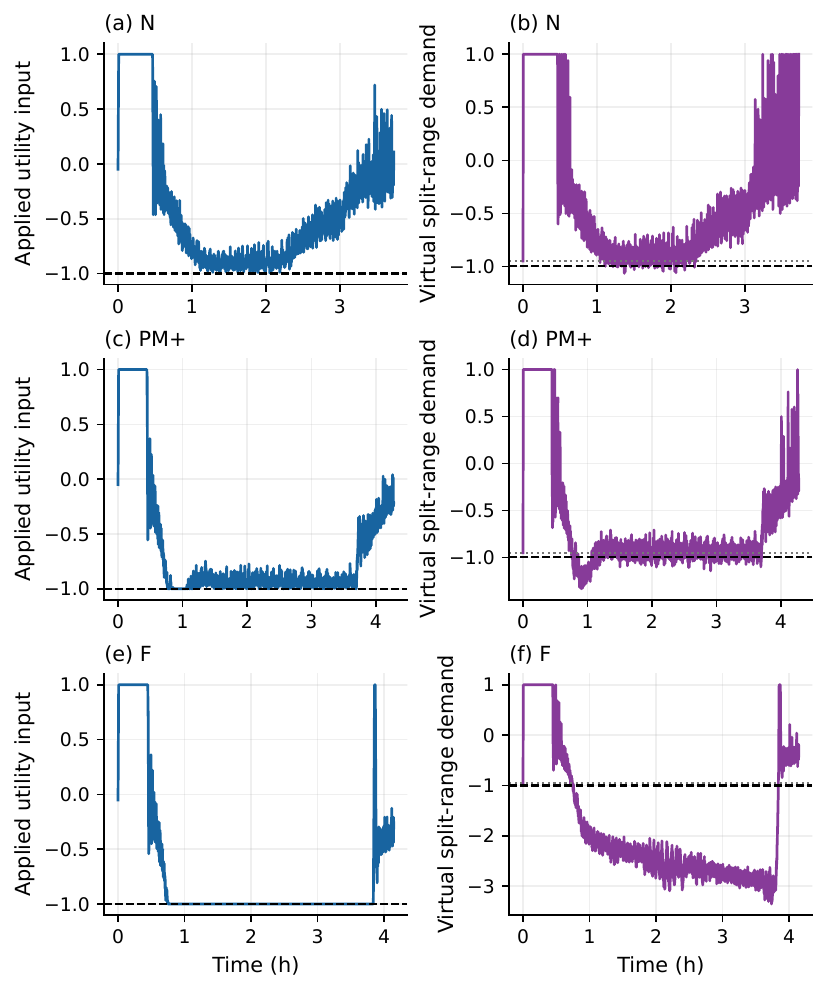}
\caption{ARC utility input (left) and virtual split-range demand (right), seed 0. Negative values request cooling; dashed lines mark maximum cooling ($-1$), and dotted lines the cooling demand setpoint ($-0.95$).}
\label{supp-fig:supp-utility}
\end{figure}

\clearpage
\section{Industrial response bound check}
\label{supp-sec:supp-industrial-tuning}
This check bounds initiator feed withdrawal and recovery using the implemented
integral projection, saturation, sample order and actuator dynamics.

\paragraph{Controller and operating window}
Let $h=1$ s, $K>0$, and $T_I>0$. In this section the dimensionless error is
$e_{{\rm VPC},n}=v_{{\rm sr},n}^{\rm v}+0.95$, and $r_+=1$, $r_-=1.5$ are the industrial integral
limits. At update $n\ge1$, the integral state and economic command are
\begin{align}
 r_n&=\operatorname{clip}(r_{n-1}+Kh e_{{\rm VPC},n}/T_I,-r_-,r_+),\label{supp-eq:supp-industrial-pi}\\
 \alpha_{{\rm eco},n}&=\operatorname{clip}(Ke_{{\rm VPC},n}+r_n,-40,10).
\end{align}
The two command maps are
\begin{align}
 u_{B,n}^{\rm cmd}&=\operatorname{clip}(10^{-7}\alpha_{{\rm eco},n},0,4\times10^{-7})\quad\mathrm{kg/s},\\
 P_{{\rm sp},n}&=\operatorname{clip}(1.5\times10^6+10^4\min(\alpha_{{\rm eco},n},0),
                       1.1\times10^6,1.5\times10^6)\quad\mathrm{Pa}.
\end{align}
Assume that the temperature cascade and both economic channels are enabled
and the integral state updates throughout the window. The finite-window bounds
end when a dose limit holds the integral state. A metered stop can only strengthen
withdrawal; a lower bound on recovery additionally requires enough remaining
initiator dose to avoid clipping. All demand inequalities below concern the
measured, filtered controller signal, rather than a noise-free thermal state.

\paragraph{Worst-case integral and command envelopes}
If $e_{{\rm VPC},j}\le-\eta_s$ throughout the first $n$ updates, then
\begin{equation}\label{supp-eq:supp-industrial-withdraw}
 r_n\le R_n^s:=\max\{-r_-,r_+-nKh\eta_s/T_I\},\qquad
 \alpha_{{\rm eco},n}\le V_n^s:=\operatorname{clip}(-K\eta_s+R_n^s,-40,10).
\end{equation}
If instead $e_{{\rm VPC},j}\ge\eta_e$, then
\begin{equation}\label{supp-eq:supp-industrial-recover}
 r_n\ge R_n^e:=\min\{r_+,-r_-+nKh\eta_e/T_I\},\qquad
 \alpha_{{\rm eco},n}\ge V_n^e:=\operatorname{clip}(K\eta_e+R_n^e,-40,10).
\end{equation}
Here $\eta_s,\eta_e>0$ are bounds on the dimensionless cooling demand error.
To prove the bounds, start with
$-r_-\le r_0\le r_+$ and apply the monotone clipped update inductively,
replacing each error by its worst allowed constant value. The affine
proportional term and command clipping preserve the order. Constant boundary
error with the corresponding extreme initial integral attains each envelope.
Both physical command maps are nondecreasing, so substitution gives feed and
pressure setpoint bounds directly from the controller law.

In particular, $K\eta_s\ge r_+$ is sufficient to remove the initiator request
at the first overload update, independently of integral state. Requiring a
pressure setpoint reduction $\Delta P_{\rm sp}$ by update $N$ gives the computable test
\begin{equation}\label{supp-eq:supp-pressure-deadline}
 K\eta_s-\max\{-r_-,r_+-NKh\eta_s/T_I\}\ge\Delta P_{\rm sp}/10^4,
 \qquad 0<\Delta P_{\rm sp}\le4\times10^5\ \mathrm{Pa}.
\end{equation}
The pressure PI uses the updated pressure setpoint on the following sample. We therefore
report the conservative bound $(N+1)h$, including that execution delay.
This is a pressure setpoint deadline: actual vessel pressure and monomer flow have
their own dynamics and no vent is modeled.

\paragraph{Actual feed and residual dose}
With $\tau_a=5$ s and rate limit $2\times10^{-8}$ kg/s$^2$, propagate the
command envelopes through the implemented recurrence
\begin{equation}\label{supp-eq:supp-actuator-envelope}
 B_n=\operatorname{clip}\!\left(B_{n-1}+\operatorname{clip}\!\left(
 \frac{h}{\tau_a}(u_{B,n}^{\rm cmd}-B_{n-1}),
 -2\times10^{-8}h,2\times10^{-8}h\right),0,4\times10^{-7}\right).
\end{equation}
Use $B_0^s=4\times10^{-7}$ kg/s for the upper withdrawal envelope and
$B_0^e=0$ for the lower recovery envelope. For $0<h/\tau_a\le1$, this map is
nondecreasing in both previous feed and command: away from rate saturation
its slopes are $1-h/\tau_a$ and $h/\tau_a$, and at rate saturation they are
1 and 0. Continuity and final clipping preserve monotonicity. Induction
therefore proves $u_{B,n}^{\rm act}\le B_n^s$ in the overload window and
$u_{B,n}^{\rm act}\ge B_n^e$ in the recovery window without dose clipping.
Consequently $h\sum_n B_n^s$ bounds the additional initiator dose, and the
first $N$ with $B_N^e\ge u_B^{\rm req}$ supplies an actual-feed recovery deadline.

For a sustained zero request from maximum feed, the first 15 steps are
rate limited and reach $10^{-7}$ kg/s. Subsequent feed decays geometrically
with ratio 0.8. It is below 1\% of maximum after 30 s. The infinite
right-endpoint dose sum is $4\times10^{-6}$ kg (4 mg). This includes the
actuator's residual feed; it does not bound heat subsequently released by
reactants already in the vessel.

\paragraph{Response requirements and gain selection}
Before running the additional batches, we specified sustained demand
$v_{\rm sr}^{\rm v}\le-1.10$ for withdrawal and $v_{\rm sr}^{\rm v}\ge-0.70$ for recovery,
giving $\eta_s=0.15$ and $\eta_e=0.25$. The illustrative commissioning
requirements are: zero initiator request at the first overload update,
a 10 kPa pressure setpoint reduction used within 30 s, and actual initiator recovery
to half its maximum within 60 s. 

For plant commissioning, the overload threshold and feed reduction time
are selected from the available thermal margin, measurement variability
and actuator response. The required pressure setpoint change depends on its
effect on monomer supply, while the acceptable underuse of cooling determines
the recovery target. Here $-1.10$ and $-0.70$ represent 10\% demand beyond
physical cooling capacity and 30\% unused capacity. The 10 kPa change lies
within the implemented 400 kPa pressure setpoint range. 

The first requirement gives $K\ge6.667$. With $K=8$, the pressure requirement
allows at most $N=29$ integral updates and gives $T_I\le43.5$ s. Direct
actuator propagation checks the recovery requirement. Thus $(8,40\,\mathrm{s})$
meets the response requirements, whereas $(8,80\,\mathrm{s})$ misses
the chosen deadlines. This calculation checks the retained settings and
illustrates how to select gains from specified response requirements.
\begin{table}[!ht]\centering\small
\caption{Industrial response bounds. Overload/recovery requires sustained virtual split-range demand $\le-1.10$/$\ge-0.70$. Bounds include all admissible initial integral and feed states.}
\label{supp-tab:supp-industrial-response}
\begin{tabular}{rrrrcc}\toprule
$K$ & $T_I$ (s) & $\Delta P_{\rm sp}\ge10$ kPa (s) & $u_B\ge u_B^{\max}/2$ (s) & Zero feed request & Pass\\\midrule
6 & 40 & 50 & 58 & No & No\\
8 & 80 & 55 & 65 & Yes & No\\
8 & 40 & 28 & 35 & Yes & Yes\\
10 & 40 & 15 & 21 & Yes & Yes\\
\bottomrule\end{tabular}
\end{table}

\paragraph{Complete-batch simulations.}
Four candidates were declared before the new outcomes were inspected and
tested in N, PM$_+$ and F with paired noise seed 301. All other controller,
filter, actuator and process settings were retained. The cooling demand setpoint
remained $-0.95$.
Table~\ref{supp-tab:supp-industrial-design} reports all twelve outcomes. These
additional design batches are reported separately from the main comparison.

\begin{table}[!ht]\centering\small
\caption{Additional design batches, seed 301. Conforming denotes completion within the $351.0\pm0.7$ K hold recipe and operating limits; -- denotes failure to complete.}
\label{supp-tab:supp-industrial-design}
\begin{tabular}{rrlrrc}\toprule
$K$ & $T_I$ (s) & Scenario & Batch time (h) & Peak $T$ (K) & Conforming\\\midrule
6 & 40 & N & 3.7367 & 351.154 & Yes\\
6 & 40 & PM$_+$ & 4.2839 & 351.206 & Yes\\
6 & 40 & F & 4.1406 & 351.821 & No\\
8 & 80 & N & 3.7308 & 351.161 & Yes\\
8 & 80 & PM$_+$ & 4.2761 & 351.210 & Yes\\
8 & 80 & F & 4.1450 & 351.249 & Yes\\
8 & 40 & N & 3.7308 & 351.161 & Yes\\
8 & 40 & PM$_+$ & 4.2789 & 351.198 & Yes\\
8 & 40 & F & 4.1453 & 351.232 & Yes\\
10 & 40 & N & 3.7283 & 351.163 & Yes\\
10 & 40 & PM$_+$ & 4.2772 & 351.256 & Yes\\
10 & 40 & F & -- & 373.155 & No\\
\bottomrule\end{tabular}
\end{table}

The retained $(8,40~\mathrm{s})$ meets both the response requirements and the
batch criteria. Setting $(6,40~\mathrm{s})$ misses the immediate feed withdrawal
and pressure setpoint requirements and reaches 351.821 K in F. The slower $(8,80~\mathrm{s})$
passes all three sampled batches but misses the chosen response deadlines.
Conversely, $(10,40~\mathrm{s})$ passes the industrial response bound check yet develops large
pressure setpoint and temperature oscillations in F and reaches the 373.15 K
trip after 1.5567 h. The response bounds assess feed reduction and recovery times,
while the batch simulations reveal the effect of loop interactions. Figure~\ref{supp-fig:supp-industrial-tuning-fault} compares these
responses with the retained setting.

The saved trajectories were checked against the integral and actuator updates,
command limits, dose metering and prescribed measurement noise. Separately
evaluated plant state updates also agreed at selected steps. For the failed F trajectory, two
half-steps with tighter tolerances were evaluated from each saved state on
4500--5603 s; the maximum one-step temperature discrepancy was
$7.81\times10^{-6}$ K. This checks the accuracy of each local integration step
along the saved trajectory. The analytic envelopes
were additionally checked against 72000 independently implemented random
comparison steps and the hand-computed zero-request actuator response.

\begin{figure}[!ht]\centering
\includegraphics[width=\textwidth]{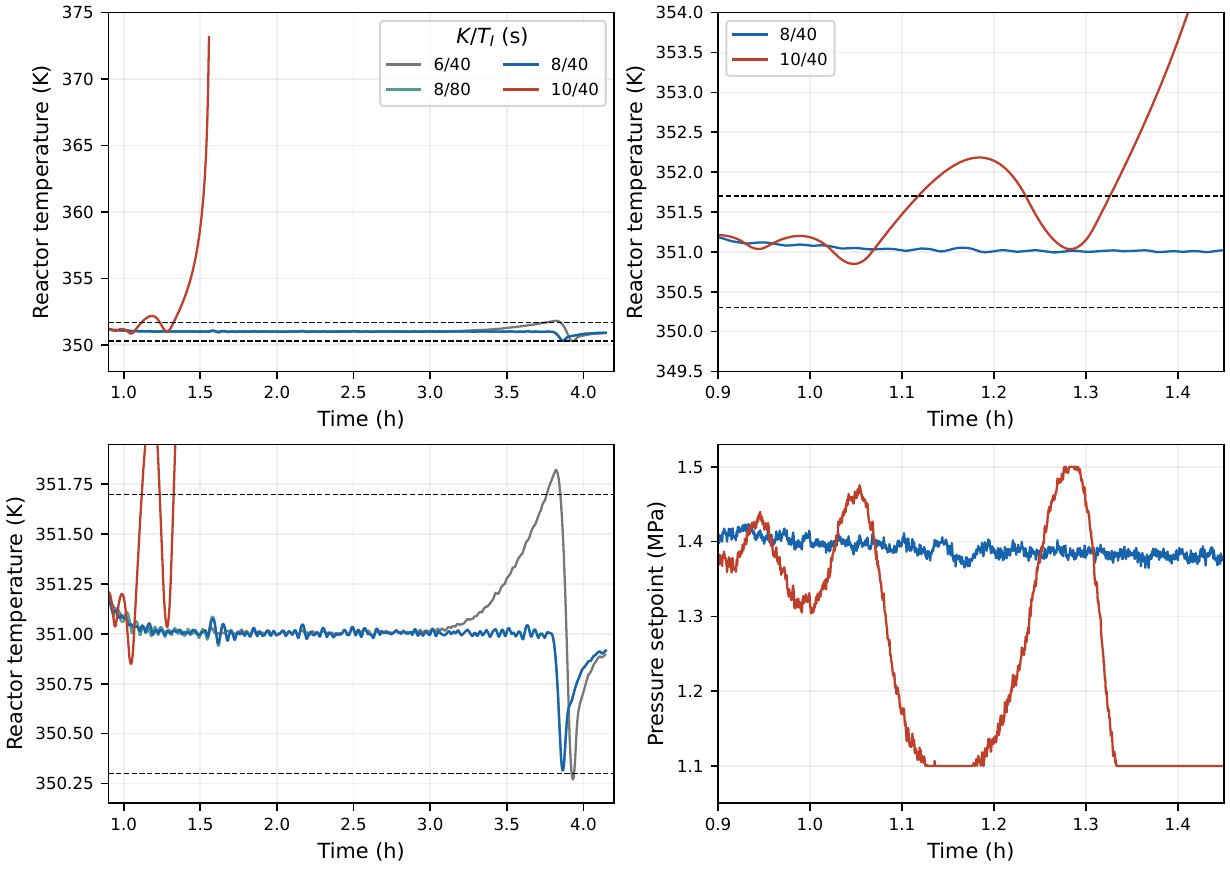}
\caption{Tuning trajectories in F, seed 301. Left: hold temperature and enlarged recipe view. Right: temperature and pressure setpoints for $K=8,10$ and $T_I=40$ s. Dashed lines: recipe bounds.}
\label{supp-fig:supp-industrial-tuning-fault}
\end{figure}

\section{Feed channel ablations}
Full ARC uses initiator feedback and pressure setpoint relaxation. The three ablations respectively hold pressure fixed, replace initiator feedback by the benchmark recipe, or remove both feedback channels. All use paired seed 0 in N, PM$_+$ and F with the remaining settings retained. The recipe requests $4\times10^{-7}$ kg/s until 0.0012 kg has been added and $10^{-7}$ kg/s thereafter. It requests zero below 1.4 MPa, after 2950 kg monomer has been charged, or at the 0.002 kg initiator cap. Table~\ref{supp-tab:ablation_mechanism} reports every run; the main text explains the channel roles.

\begin{table}[!t]\centering\footnotesize
\caption{ARC ablations, paired seed 0. Time is batch completion time for completed runs and stopping time otherwise. Peak temperature uses the full true trajectory. Pressure setpoint relaxation is measured from 3300 s until monomer charge completion.}\label{supp-tab:ablation_mechanism}
\setlength{\tabcolsep}{3pt}
\begin{tabular}{lllccc}\toprule
Scen. & Variant & Outcome & Time (h) & Peak $T$ (K) & Relax. (\%)\\\midrule
N & Full ARC & Complete & 3.731 & 351.166 & 0.1\\
N & Fixed $P_{\rm sp}$ & Complete & 3.731 & 351.166 & 0.0\\
N & Recipe $u_B$ & Complete & 4.141 & 351.126 & 0.0\\
N & Recipe only & Complete & 4.141 & 351.126 & 0.0\\
PM$_+$ & Full ARC & Complete & 4.278 & 351.192 & 22.1\\
PM$_+$ & Fixed $P_{\rm sp}$ & Trip & 0.985 & 373.212 & 0.0\\
PM$_+$ & Recipe $u_B$ & Complete & 4.505 & 351.222 & 95.5\\
PM$_+$ & Recipe only & Trip & 1.246 & 373.406 & 0.0\\
F & Full ARC & Complete & 4.145 & 351.227 & 100.0\\
F & Fixed $P_{\rm sp}$ & Trip & 0.921 & 373.621 & 0.0\\
F & Recipe $u_B$ & Complete & 4.216 & 353.387 & 97.2\\
F & Recipe only & Trip & 1.175 & 373.336 & 0.0\\
\bottomrule\end{tabular}
\par\smallskip\begin{minipage}{\textwidth}\footnotesize A completed batch can still exceed the 351.7 K operating limit; the protective trip is 373.15 K. Relax. is the fraction of this interval with the pressure setpoint more than 1 Pa below 1.5 MPa; 1 Pa is the numerical tolerance. An unavailable interval is shown as --.\end{minipage}
\end{table}
\section{Additional tuning and sensitivity checks}
\label{supp-sec:supp-additional-checks}
Three PM$_+$ tuning settings use $(|v_{{\rm sr},sp}|,K_{\rm VPC})=(0.90,6),(0.95,8),(0.98,10)$ with design seed 94. The setpoint and gain change together, so the comparison assesses three combined settings. On $t\ge3300$ s before monomer charge completion, cooling demand overload means $v_{\rm sr}^{\rm v}\le-1$; pressure setpoint relaxation means a setpoint more than 1 Pa below 1.5 MPa, using the same numerical tolerance as the ablation table. We report the fraction of this interval for which each condition holds.

Two paired stress cases use N and PM$_+$ with seed 50. The noise standard deviations are multiplied by three and all actuator time constants by two. ARC and A-NMPC retain their evaluation settings, including the original estimator covariances.

In Table~\ref{supp-tab:pmplus_three_designs}, the baseline completes 188 s earlier than the conservative setting and 72 s later than the aggressive setting. Faster completion is accompanied by more frequent cooling overload and pressure setpoint relaxation.

\begin{table}[!t]\centering\footnotesize
\caption{Three PM$_+$ tuning settings, design seed 94.}\label{supp-tab:pmplus_three_designs}
\setlength{\tabcolsep}{3pt}
\begin{tabular}{lccccc}\toprule
Design & $t_f$ (h) & $T_{\rm peak}$ (K) & $\min v^{\rm v}_{\rm sr}$ & Overload (\%) & Relax. (\%)\\\midrule
Conservative (0.90,6) & 4.328 & 351.177 & -1.117 & 3.3 & 14.6\\
Baseline (0.95,8) & 4.276 & 351.203 & -1.339 & 18.5 & 19.6\\
Aggressive (0.98,10) & 4.256 & 351.222 & -1.417 & 34.7 & 24.9\\
\bottomrule\end{tabular}

\end{table}

Table~\ref{supp-tab:stress-revision} reports operation with higher noise and slower actuators than in the main comparison. Temperature exceedance is observed in A-NMPC under N (peak 351.738 K), ARC under PM$_+$ (peak 352.082 K), and A-NMPC under PM$_+$ (peak 352.015 K). These exceedances show why commissioning must account for measurement noise and actuator response when assessing the controller settings.

\begin{table}[!t]\centering\footnotesize
\caption{Synthetic sensitivity checks, seed 50: threefold measurement noise and doubled actuator time constants, with unchanged controller and estimator settings.}\label{supp-tab:stress-revision}
\setlength{\tabcolsep}{3pt}
\begin{tabular}{lllcccc}\toprule
Scen. & Controller & Outcome & Spec. & Time (h) & Peak $T$ (K) & Exceed. (\%)\\\midrule
N & ARC & completed & Yes & 3.750 & 351.166 & 0.00\\
N & A-NMPC & completed & No & 3.754 & 351.738 & 1.11\\
PM$_+$ & ARC & completed & No & 4.456 & 352.082 & 7.76\\
PM$_+$ & A-NMPC & completed & No & 4.336 & 352.015 & 10.71\\
\bottomrule\end{tabular}
\par\smallskip\begin{minipage}{\textwidth}\footnotesize Spec. requires physical completion and conformity with the two-sided hold recipe and operating upper bounds. Time is the recorded stopping time. Exceedance uses $T_r>351.7$ K over all one-second right-endpoint samples.\end{minipage}
\end{table}

\FloatBarrier

\end{document}